\documentclass[conference,compsoc]{IEEEtran}
\ifCLASSOPTIONcompsoc
  \usepackage[nocompress]{cite}
\else
  \usepackage{cite}
\fi
\ifCLASSINFOpdf
\else
\fi
\usepackage{subcaption}
\usepackage{bm}
\usepackage{bbding}
\usepackage{listings}
\usepackage{makecell}
\usepackage{amsthm}
\usepackage{xcolor}
\usepackage{mathtools}
\usepackage{tikz}
\usepackage[utf8]{inputenc}
\usepackage{amsmath,amsfonts}
\makeatletter
\@ifclassloaded{acmart}{}{\usepackage{amssymb}}
\makeatother
\usepackage{algorithmic}
\usepackage{algorithm}
\usepackage{array}
\usepackage{textcomp}
\usepackage{stfloats}
\usepackage{url}
\usepackage{verbatim}
\usepackage{enumitem} 
\usepackage{graphicx}
\usepackage{xspace}
\usepackage{tabularray}
\usepackage{booktabs}
\usepackage{multirow}
\usepackage{threeparttable}
\usepackage{diagbox}
\usepackage{soul}
\usepackage{hyperref}

\providecommand{\Description}[2][]{}

\newcommand{\sysname}{\textsc{SafeStats}\xspace}

\usepackage{amsthm}
\newtheorem{definition}{Definition}

\newtheorem{theorem}{Theorem}
\newtheorem{claim}{Claim}

\newtheorem{lemma}{Lemma}

\theoremstyle{definition}

\newcolumntype{C}{>{$}c<{$}}

\newcommand{\myherebox}[6]{
	\begin{figure}[ht!]
		\centering
		\begin{tikzpicture}
		\node[anchor=text,text width=\columnwidth-1.1cm, draw, rounded corners, line width=1pt, fill=#3, inner sep=5mm, font=\fontsize{8}{10}\selectfont] (big) {\\#4};
		\node[draw, rounded corners, line width=.5pt, fill=#2, anchor=west, xshift=5mm] (small) at (big.north west) {#1};
		\end{tikzpicture}
		\par\vspace{0.4em}
        \caption{#5}
        \label{#6}
        \Description[Short description]{}
	\end{figure}
}

\begin{document}
%
\title{\sysname: Efficient 2PC Protocols for Data Statistic-Related Functions}

\author{
\IEEEauthorblockN{Tanren Liu}
\IEEEauthorblockA{Xidian University\\
Xi'an, Shaanxi, China\\
fatw\_email@163.com
}
\and
\IEEEauthorblockN{Xianjia Meng}
\IEEEauthorblockA{Northwest University\\
Xi'an, Shaanxi, China\\
xianjiam@nwu.edu.cn
}
\and
\IEEEauthorblockN{Yang Liu}
\IEEEauthorblockA{Xidian University\\
Xi'an, Shaanxi, China\\
bcds2018@foxmail.com
}
\and
\IEEEauthorblockN{Xin Kang}
\IEEEauthorblockA{Xidian University\\
Xi'an, Shaanxi, China\\
kangxin@stu.xidian.edu.cn
}
\and
\IEEEauthorblockN{Chenhui You}
\IEEEauthorblockA{Xidian University\\
Xi'an, Shaanxi, China\\
25151213660@stu.xidian.edu.cn
}
\and
\IEEEauthorblockN{Yong Zeng}
\IEEEauthorblockA{Xidian University\\
Xi'an, Shaanxi, China\\
yzeng@mail.xidian.edu.cn
}
\and
\IEEEauthorblockN{Zhuo Ma}
\IEEEauthorblockA{Xidian University\\
Xi'an, Shaanxi, China\\
mazhuo@mail.xidian.edu.cn
}}


%


\maketitle

\begin{abstract}
Statistical analysis on sensitive datasets like medical records and financial transactions is essential for decision-making, but raises significant privacy concerns. 
While existing secure Two-Party Computation (2PC) makes extensive efforts in designing the common secure primitives (e.g., addition and multiplication) or machine learning-related functions, few pay attention to the statistical functions.
In this paper, we propose \sysname, a secure toolkit tailored for 2PC secure statistical analysis.
Specifically, to develop \sysname, we first refer to Microsoft Excel’s
statistical library and summarize that most statistical operations can be achieved with three core functions: 1) frequency counting, 2) sorting, and 3) non-linear math functions.
Then, for each core statistical function, \sysname presents an efficient 2PC implementation.
For secure frequency counting, \sysname adopts a secure shift-based strategy to avoid invoking expensive 2PC equality test protocols.
For secure sort, \sysname involves a secure segment-indicator protocol to achieve secure counting-based sort, which enables fast element sorting over specific statistical scenarios without the need for secure comparison.
For non-linear math functions, we enhance the current reduce-then-approximate paradigm by introducing a bisection-based range reduction protocol.
Finally, we implement \sysname and test it on $14$ common statistical analysis cases.
As an example, for the chi-square test, \sysname achieves a $1.5\times$ runtime speedup and a $4.2\times$ reduction in communication compared to directly using the current general-purpose 2PC library to realize it.

\end{abstract}


%
\IEEEpeerreviewmaketitle

\section{Introduction}

Statistical analysis plays a pivotal role in enabling data-driven decision-making across diverse domains, including big data analytics, social sciences, and etc. 
However, in these domains, the pervasive issue of data privacy leakage has emerged as a significant concern for both industry and academia \cite{seyidova2025mathematical, pandey2025privacy, liu2024survey}.
To address this challenge, the secure two-party computation (2PC) technique has gained considerable attention from the research community, offering a promising avenue for mitigating privacy risks.

Although 2PC is extensively studied~\cite{guo2025seaf,liu2025mlformer,li2024nimbus, hou2023ciphergpt, chandran2019ezpc}, significant challenges remain in deploying it within the statistical analysis scenarios.
Existing studies on 2PC primarily focus on either low-level secure primitives like secure multiplication and secure comparison~\cite{eqlu2024efficient,rathee2020cryptflow2,huang2022cheetah}, or high-level applications such as privacy-preserving machine learning model training and predictions~\cite{xu2025breaking, guo2025seaf, luo2024secformer}.
Nevertheless, turning to the statistical operations such as $\mathsf{Median}$ and $\mathsf{Rank}$, existing schemes~\cite{hou2023ciphergpt, hamada2012practically} typically adopt the most straightforward solution: directly utilizing 2PC primitives to reproduce the operations.
As a result, even with the state-of-the-art (SOTA) protocols~\cite{hou2023ciphergpt, bogdanov2014practical}, computing a $\mathsf{Rank}$ function over $2^{12}$ elements still requires more than $18$ minutes.
Intuitively, such overhead renders these solutions impractical for most applications.

\subsection{Our Techniques}
In this work, we present \sysname, a secure toolkit tailored for 2PC privacy-preserving statistical analysis.
Specifically, to develop \sysname, our research mainly proceeds with the following three steps: 
1) analyze the requirement to realize 2PC statistical analysis;
2) propose three secure and efficient 2PC protocols for core statistical functions;
3) test \sysname over an extensive range of specific statistical analysis cases.

\begin{table}[htbp]
  \centering
  \begin{threeparttable}
    \caption{Commonly used statistical functions in Excel and their required building blocks:
    counting function ($\mathsf{Frequency}$), 
    sort function ($\mathsf{Sort}$),
    mathematical non-linear functions ($\ln x, \sqrt{x}, 1/x$).
    Operations relying only on addition/multiplication (e.g., mean, variance) 
    are excluded since these are well-supported in standard 2PC frameworks.}
\label{tab: function_analysis}
\setlength{\tabcolsep}{1.6pt}
\renewcommand{\arraystretch}{1.2}
\begin{tabular}{cc|ccc}
\hline
\multicolumn{2}{c|}{\multirow{2}{*}{Statistical Operations}} & \multicolumn{3}{c}{Building Blocks} \\ \cline{3-5} 
\multicolumn{2}{c|}{} & \multicolumn{1}{c|}{Counting Func.} & \multicolumn{1}{c|}{Sort Func.} & Math Func(s). \\ \hline
\multicolumn{1}{c|}{\multirow{5}{*}{\begin{tabular}[c]{@{}c@{}}Order \&\\ Frequency\end{tabular}}} & Median & \multicolumn{1}{c|}{$\square$} & \multicolumn{1}{c|}{$\blacksquare$} & $\square$ \\ \cline{2-5} 
\multicolumn{1}{c|}{} & Mode & \multicolumn{1}{c|}{$\blacksquare$} & \multicolumn{1}{c|}{$\blacksquare$} & $\square$ \\ \cline{2-5} 
\multicolumn{1}{c|}{} & Percentile & \multicolumn{1}{c|}{$\square$} & \multicolumn{1}{c|}{$\blacksquare$} & $\square$ \\ \cline{2-5} 
\multicolumn{1}{c|}{} & Rank & \multicolumn{1}{c|}{$\square$} & \multicolumn{1}{c|}{$\blacksquare$} & $\square$ \\ \cline{2-5} 
\multicolumn{1}{c|}{} & Min/Max & \multicolumn{1}{c|}{$\square$} & \multicolumn{1}{c|}{$\blacksquare$} & $\square$ \\ \hline
\multicolumn{1}{c|}{\multirow{3}{*}{\begin{tabular}[c]{@{}c@{}}Descriptive\\ Statistics\end{tabular}}} & Std & \multicolumn{1}{c|}{$\square$} & \multicolumn{1}{c|}{$\square$} & $\blacksquare$ \\ \cline{2-5} 
\multicolumn{1}{c|}{} & Skew & \multicolumn{1}{c|}{$\square$} & \multicolumn{1}{c|}{$\square$} & $\blacksquare$ \\ \cline{2-5} 
\multicolumn{1}{c|}{} & Kurt & \multicolumn{1}{c|}{$\square$} & \multicolumn{1}{c|}{$\square$} & $\blacksquare$ \\ \hline
\multicolumn{1}{c|}{\multirow{2}{*}{\begin{tabular}[c]{@{}c@{}}Regression\\ Analysis\end{tabular}}} & Linest & \multicolumn{1}{c|}{$\square$} & \multicolumn{1}{c|}{$\square$} & $\blacksquare$ \\ \cline{2-5} 
\multicolumn{1}{c|}{} & Logest & \multicolumn{1}{c|}{$\square$} & \multicolumn{1}{c|}{$\square$} & $\blacksquare$ \\ \hline
\multicolumn{1}{c|}{Correlation} & Pearson & \multicolumn{1}{c|}{$\square$} & \multicolumn{1}{c|}{$\square$} & $\blacksquare$ \\ \hline
\multicolumn{1}{c|}{\multirow{3}{*}{\begin{tabular}[c]{@{}c@{}}Hypothesis\\ Testing\end{tabular}}} & T.Test & \multicolumn{1}{c|}{$\square$} & \multicolumn{1}{c|}{$\square$} & $\blacksquare$ \\ \cline{2-5} 
\multicolumn{1}{c|}{} & CHISQ.Test & \multicolumn{1}{c|}{$\blacksquare$} & \multicolumn{1}{c|}{$\square$} & $\blacksquare$ \\ \cline{2-5} 
\multicolumn{1}{c|}{} & F.Test & \multicolumn{1}{c|}{$\square$} & \multicolumn{1}{c|}{$\square$} & $\blacksquare$ \\ \hline
\end{tabular}
    \begin{tablenotes}[flushleft]
      \footnotesize
      \item $\blacksquare$ = required; $\square$ = not required.
    \end{tablenotes}
  \end{threeparttable}
\end{table}

\subsubsection{Requirement analysis}
Statistical analysis involves a wide range of different operations, which makes it impossible to enumerate and develop a 2PC protocol for each of them.
Therefore, to simplify the design of \sysname, we refer to Microsoft Excel’s statistical library\footnote{https://support.microsoft.com/en-us/office/statistical-functions-reference-624dac86-a375-4435-bc25-76d659719ffd} and summarize the common functions used for most statistical operations.
\autoref{tab: function_analysis} shows our investigation result.
From this, we can categorize the common statistical operations into three groups:
\begin{itemize}[leftmargin=1.5em]
    \item The frequency counting-based operations ($\mathsf{Frequency}$), which outputs the frequency of the chosen element values and offers support to counting-related statistical operations, e.g., $\mathsf{Mode}$ and $\mathsf{CHISQ.Test}$.
    \item The sorting-based operations ($\mathsf{Sort}$), which sort the chosen data elements with a specific order and offer support to order-related operations, e.g., $\mathsf{Median}$ and $\mathsf{Rank}$.
    \item Non-linear math function(s)-based operations, e.g. $\mathsf{Logarithm}$, used to implement algebraic statistical functions, e.g., $\mathsf{Std}$ and $\mathsf{Skew}$.
\end{itemize}
Here, we omit the reference to basic linear functions like addition and multiplication because they have been extensively discussed in prior works.
According to the above, we can derive that to develop a secure 2PC library for private statistical analysis, the main task is to implement three core functions, i.e., \textbf{secure counting function}, \textbf{secure sort function}, and \textbf{secure math non-linear function(s)}.

\subsubsection{Secure counting function (\autoref{text: count_method})}
To obtain the private frequency counting of each value $v$, previous methods~\cite{zhang2011generic} need to separately check every input whether it is equal to $v$, leading to a high execution complexity.
\sysname introduces a new paradigm for frequency counting without the dependence on equality test.
The key observation of our paradigm is that the result of performing equality test between input $x$ and each value $v\in[0, m-1]$ can be viewed as a boolean vector of length $m$. 
The vector has value $1$ at position $x$ and $0$ elsewhere, abbr. $\vec{e}^x$.
Instead of constructing $\vec{e}^{x}$ through $m$ equality tests, we note that $\vec{e}^{x}$ can be derived as an $x$ offset right-shifted version of vector $\vec{e}^{0}$.
With this shift-based representation, we can replace the $m$ times equality test with one vector shift to implement the secure counting function.
Note that our oblivious shift only needs $\mathcal{O}(\log m)$ invocation of oblivious transfer (OT), while equality test needs $\mathcal{O}(m\ell)$ OTs~\cite{rathee2020cryptflow2, huang2022cheetah}, where $\ell$ is the data length of the input.
Concretely, for frequency counting, \sysname achieves $4.0$-$7.9\times$ runtime improvement and $3.1$-$8.6 \times$ communication improvement compared to the current equality-based frequency counting methods~\cite{zhang2011generic}.

\subsubsection{Secure sort function (\autoref{text: order_method})}
Prior secure sort~\cite{wprfpeceny2025efficient, hou2023ciphergpt, hamada2012practically} offer efficient general-purpose solutions.
However, they become inefficient for statistical tasks on large dataset with small, enumerable feature domains\footnote{We say that such an assumption is common for statistical apps, e.g., analyzing the relationship between income level and educational attainment at the city level, where each feature has fewer than 10 discrete categories, while the dataset contains tens of thousands of records.}, 
because their performance is dominated by secure comparison, which is widely recognized as the performance bottleneck in 2PC~\cite{ABYv2,huang2022cheetah, zhou2023bicoptor}.
To break the bottleneck, \sysname proposes a secure counting-based sort protocol without the reliance on secure comparisons. 
Commonly, counting sort is carried out with two steps: 1) frequency vector construction and 2) element insertion.
Intuitively, the former step can be well implemented with the above-mentioned secure frequency counting protocol.
While for the second step, a direct solution is to use the secure iterative insertion method introduced by~\cite{azogagh2024non}.
However, as depicted in~\autoref{fig: count_sort}, this insertion method can only be executed sequentially, which makes it suffer from a very high circuit depth (i.e., $\mathcal{O}(n)$).

Relatively, \sysname introduces a secure segment-indicator protocol to achieve elements insertion for counting-based sort.
With private input $a, b$, this protocol outputs a secret-shared indicator vector equal to $1$ on the interval $[a, b]$ and $0$ elsewhere
Using this protocol, \sysname can insert all elements with the same value into the output vector via a single element-wise multiplication as shown in~\autoref{fig: count_sort}.
As a result, besides the $\mathcal{O}(1)$ circuit depth, \sysname only needs $\mathcal{O}(n(m + \log m))$ OTs to achieve secure sort, which outperforms the previous comparison-based sort with $\mathcal{O}(n\log n\log (nm))$ OTs, where $n$ is the data size and $m$ is the domain size.
In our experiments with $m \le \log^2 n$, \sysname achieves $3.4$-$20.5 \times$ runtime speedup and $1.3$-$7.6\times$ communication reduction compared to the prior methods~\cite{wprfpeceny2025efficient, hou2023ciphergpt, hamada2012practically}. 

\begin{figure}[htbp]
    \centering
    \includegraphics[scale=0.65]{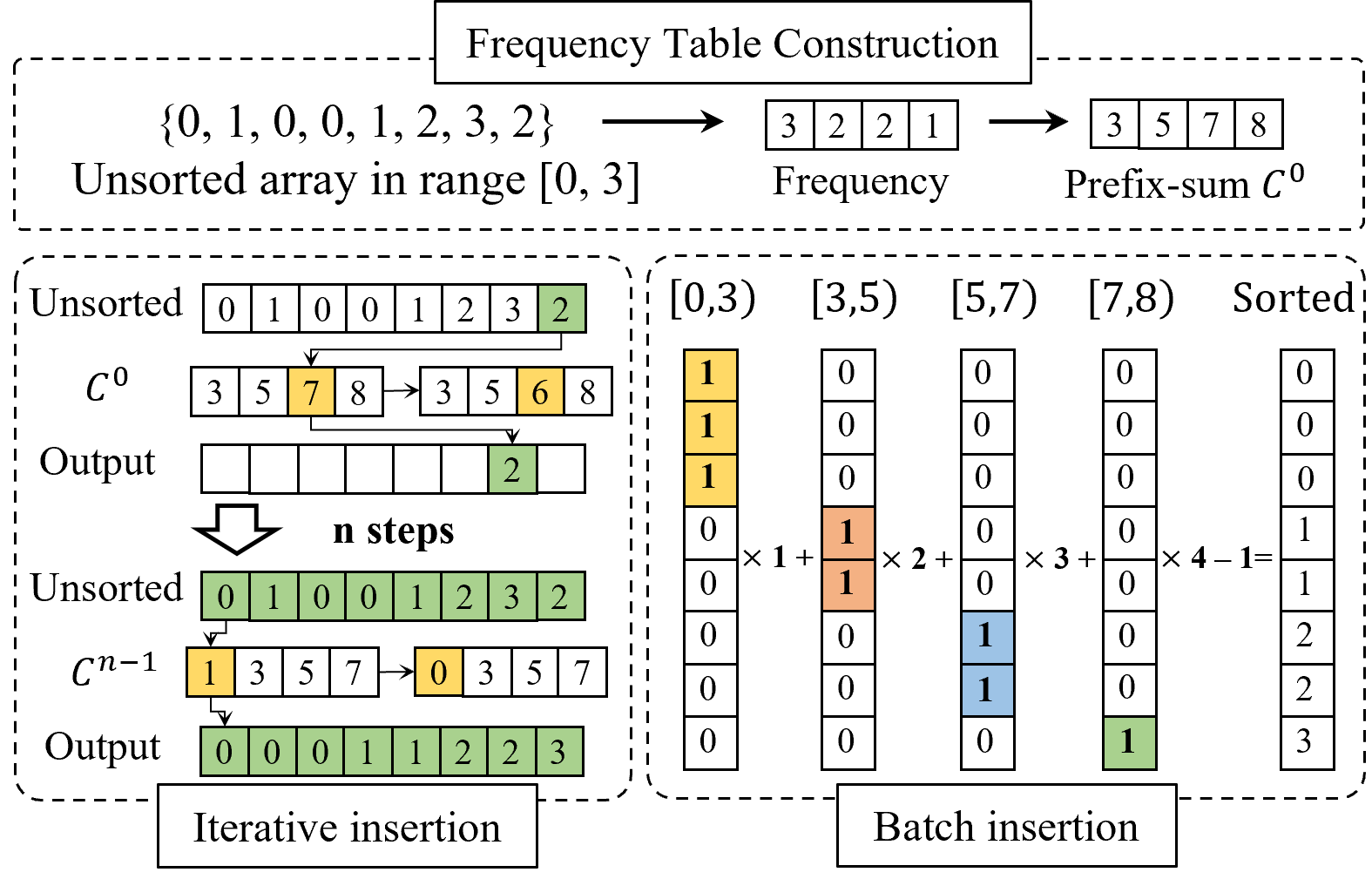}
    \caption{A toy example of iterative and batch insertion in counting sort.}
    \label{fig: count_sort}
    \Description[Short description]{Toy example for iterative insertion and batch insertion in counting sort}
\end{figure}

\subsubsection{Secure non-linear math functions (\autoref{text: math_method})}
To securely evaluate non-linear functions, existing pipelines~\cite{keller2020mp,hua2024ppglove,rathee2021sirnn,keller2022secure,lu2020faster}
mainly contains two steps: 1) range normalization, rewriting the input as $x = 2^t x'$; and 2) polynomial approximation of the non-linear function over $x'$.
A dominant cost in this pipeline arises from range reduction (40\%-60\%), whose core operation is identifying the Most Significant Non-Zero (MSNZ) bit of $x$.
Existing protocol~\cite{rathee2021sirnn} computes the MSNZ bit of $x$ with a linear strategy: decomposing the $\ell$-bit input into $d$ digits of $c$-bit ($\ell = c\cdot d$), and sequentially examining the digits from most to least to locate the MSNZ digit.
The exact MSNZ bit is then obtained via a private lookup table.
This leads to $\mathcal{O}(d)$ communication rounds.
Observing that the MSNZ digit always lies in the non-zero half of $x$, \sysname proposes a bisection-based protocol for MSNZ bit identification.
Instead of scanning digits sequentially, our protocol recursively partitions the input and selects the non-zero half at each step, reducing the communication rounds to $\mathcal{O}(\log d)$.
By incorporating this bisection-based MSNZ protocol, our implementation of \sysname reduces runtime by $1.2$-$1.7\times$ and communication by $1.15$-$1.16\times$ for secure non-linear function evaluation compared to the current work ~\cite{rathee2021sirnn}.

\subsubsection{Case Study (\autoref{text: case})}
To validate the improvement of \sysname, we experiment with it on $14$ specific cases of statistical operations.
Since parts of the operations (e.g., $\mathsf{Rank}$ and $\mathsf{CHISQ.Test}$) have not been discussed in previous work, we directly implement them with SOTA 2PC primitives~\cite{ma2023secretflow,chandran2019ezpc} as baselines.
According to the results, \sysname gives an optimization on all $14$ cases.
As an example, for the chi-square test, \sysname achieves a $1.5\times$ runtime speedup and a $4.2\times$ reduction in communication compared to baselines.

The contributions of this work are summarized below.
\begin{itemize}[leftmargin=1.5em]
    \item \textbf{Secure counting function}. 
    We propose a secure shift-based frequency counting function that reduces the OT invocations from $\mathcal{O}(m\ell)$ to $\mathcal{O}(\log m)$. 
    Moreover, we further optimize a customized secure counting function with the Chinese Remainder Theorem (CRT), which decomposes the feature domain into smaller sub-domains, enabling the construction of shorter frequency vectors.
        
    \item \textbf{Secure sort function}.    
    We introduce a novel secure counting sort protocol designed for statistical tasks on large datasets with small feature domains.
    By leveraging a secure segment-indicator protocol, our approach achieves constant communication rounds. 
    In addition, considering $\mathsf{percentiles}$, a special case of sorting, we reframe it as an interval-testing problem, achieving notable efficiency when the feature domain is small.

    \item \textbf{Secure non-linear math functions}.
    We improve the efficiency of secure non-linear functions by introducing a novel secure MSNZB protocol that reduces the communication rounds from $\mathcal{O}(n)$ to $\mathcal{O}(\log n)$.
    Furthermore, we characterize the relationship between the degree of an approximation polynomial and its valid approximation range, and offer an alternative efficiency optimization direction for secure non-linear math functions, i.e., interval compression.

    \item \textbf{Security and performance}.
    Besides providing formal security proof for our proposed protocols, we implement \sysname and benchmark the proposed protocols. 
    The results on $14$ commonly used statistical functions demonstrate that \sysname provides a more efficient implementation compared to directly using EzPC~\cite{chandran2019ezpc}.
\end{itemize}

\section{Preliminaries}

\textit{Notation.} For an vector $\vec{a}$, $\vec{a}[i]$ denotes the i-th element in $\vec{a}$. $[n]$ denotes the set of integers $\{0, \dots ,n-1\}$. 
Let $a \gets \mathbb{Z}_{2^\ell}$ denotes sampling a number uniformly at random from $\mathbb{Z}_{2^\ell}$.
Let $1\{b\}$ denotes the indicator function that is $1$ when $b$ is \textit{true} and $0$ when $b$ is \textit{false}.
Let $\ell_m = \lceil \log m\rceil$ and $\ell_n = \lceil \log n\rceil$.
$\lambda$ is the security parameter and is typically set to 128.

\subsection{System Model}
We consider a two-party system consisting of parties $P_0$ and $P_1$ that jointly conduct a statistical task, such as computing the correlation between education level and income.
$P_0$ and $P_1$ owns the secret-shared input $\langle \vec{x}\rangle_b$ respectively.
The parties compute corresponding statistical functions on the secret-shared inputs via \sysname.
The output is secret-shared between parties, which can be revealed when necessary.

\subsection{Security Model}
We provide security in the Universal Composability (UC) framework~\cite{canetti2001universally} against a static semi-honest probabilistic polynomial-time (PPT) adversary $\mathcal{A}$.
In this model, a computationally bounded adversary $\mathcal{A}$ corrupts either $P_0$ or $P_1$ at the start of the protocol execution and then honestly follows the prescribed steps of the protocol.
Security is defined by comparing two executions:
(1) a real-world execution $\mathsf{Real}_{ \Pi,\mathcal{A},\mathcal{Z}}(1^\lambda)$, where $P_0$ and $P_1$ run the protocol in the presence of $\mathcal{A}$ and an environment $\mathcal{Z}$; and
(2) an ideal-world execution $\mathsf{Ideal}_{\mathcal{F},\mathcal{S},\mathcal{Z}}(1^\lambda)$, where both parties submit their inputs to a trusted functionality that performs the computation correctly and outputs the result.
Many of our constructions employ multiple sub-protocols.
We formalize these using the hybrid model, which resembles the real-world execution except that calls to sub-protocols are replaced with invocations of their corresponding ideal functionalities.
A protocol that makes use of an ideal functionality $\mathcal{F}$ is said to operate in the $\mathcal{F}$-hybrid model.

\begin{definition}
    We say protocol $\Pi$ UC-secure realizes functionality $\mathcal{F}$ if for all PPT adversaries $\mathcal{A}$ there exists a PPT simulator $\mathcal{S}$ such that for all PPT environment $\mathcal{Z}$, it holds:
    $$
    \mathsf{Real}_{ \Pi,\mathcal{A},\mathcal{Z}}(1^\lambda) \; \approx \;  \mathsf{Ideal}_{\mathcal{F},\mathcal{S},\mathcal{Z}}(1^\lambda) 
    $$ 
\end{definition}

\subsection{Cryptographic Primitives}
\label{text: primitives}
\textbf{Secret Sharing Schemes}.
Throughout this work, we use 2-out-of-2 additive secret sharing schemes over different rings. 
There are 3 specific rings that we consider: $\mathbb{Z}_2$, $\mathbb{Z}_{2^\ell}$, and $\mathbb{Z}_N$, for a positive $N$ ($N$ is not a power of $2$).
The secret share function takes as input $x \in \mathbb{Z}_{2^\ell}$, and output shares over $\mathbb{Z}_{2^\ell}$, denoted by $\langle x\rangle_0^{2^\ell}$ and $\langle x\rangle_1^{2^\ell}$, with the only constriction is $\langle x\rangle_0^{2^\ell} + \langle x\rangle_1^{2^\ell} = x$ ($+$ denotes addition in $\mathbb{Z}_{2^\ell}$). 
The secret share schemes are defined similarly in $\mathbb{Z}_2$ and $\mathbb{Z}_N$.
We sometimes refer to shares over $\mathbb{Z}_N$ and $\mathbb{Z}_{2^\ell}$ as arithmetic shares, denoted as $\langle \cdot \rangle^{2^\ell}, \langle \cdot \rangle^N$ and shares over $\mathbb{Z}_2$ as boolean shares, denoted as $\langle \cdot \rangle^B$.
When the context is clear, we omit the superscript notation like $\langle \cdot\rangle$ for simplicity.

\textbf{Oblivious Transfer}.
Oblivious Transfer (OT) is the functionality for the sender to input messages $m_0, \dots,m_{k-1}$ and the receiver to input the choice $c \in Z_k$.
The receiver receives $m_c$ as output, but doesn't know the other messages.
The sender has nothing about the choice $c$.
We denote a 1-out-of-k OT as $(1,k)\text{-OT}_{\ell}$.
Additionally, we also use two specific OT: the 1-out-of 2 correlated OT $(1,2)\text{-COT}_\ell$ and the 1-out-of 2 random OT $(1,2)\text{-ROT}_\ell$.
$(1,2)\text{-COT}_\ell$ is defined as follows: sender inputs a correlation $x$, receiver inputs a choice bit $i \in \{0,1\}$ and the functionality outputs a random element $r$ to the sender and $-r + i\cdot x$ to the receiver.
While in $(1,2)\text{-ROT}_\ell$, the sender receives two random messages $m_1, m_2$.
The receiver receives the random choice bit $b$ and the corresponding message $m_b$.
These OTs can be implemented using either IKNP-style~\cite{asharov2013more} or VOLE-style~\cite{yang2020ferret}.
For IKNP-style OT, excluding the one-time setup cost for the base OTs, $(1,k)\text{-OT}_\ell$, $(1,2)\text{-COT}_\ell$ require $2\lambda + k\ell$, $\lambda + \ell$.
Specifically, simpler $(1,2)$-$\text{OT}_\ell$ has a communication of $\lambda + 2\ell$.
For fair comparison, we adopt IKNP-style OT for the following protocol complexity analysis.

\textbf{Random Vector Oblivious Shift Evaluation}.
The functionality $\mathcal{F}_\text{RVOSE}$~\cite{eqlu2024efficient} takes a boolean vector $\vec{a}$ as input from $P_0$, and outputs $\vec{b}$ to $P_0$, and outputs an offset $\epsilon$ and vector $\vec{c}=\text{shift}(\epsilon, \vec{a})\oplus \vec{b}$ to $P_1$.
This functionality can be implemented by~\cite{eqlu2024efficient} with a communication complexity $3\lambda \log n$.
We can freely extend this functionality $\mathcal{F}_\text{RVOSE}$ to the arithmetic version, where $\vec{c}=\text{shift}(\epsilon, \vec{a}) + \vec{b}$.
The only difference with $\Pi_\text{RVOSE}$ in~\cite{eqlu2024efficient} is that we use the pseudo-random generator to derive the puncture matrix, where each element in $\mathbb{Z}_{2^\ell}$ instead of $\mathbb{Z}_{2}$.

\textbf{Multiplexer \& B2A conversion}.
The functionality $\mathcal{F}_\text{MUX}$~\cite{rathee2021sirnn} takes as input arithmetic shares of $a$, boolean shares of $b$, and outputs shares of $a$ if $b=1$, else outputs shares of $0$.
A protocol for $\mathcal{F}_\text{MUX}$ can be implemented by $2$ simultaneous calls to 
$(1,2)\text{-COT}_\ell$ and communication complexity is $2(\lambda + \ell)$.
The functionality $\mathcal{F}_\text{B2A}$~\cite{rathee2021sirnn} takes boolean shares as input and gives out arithmetic shares of the same value as output. 
It can be realized via one call to $(1,2)\text{-COT}_\ell$, and communication complexity is $\lambda + \ell$.

\textbf{AND Gate}.
The functionality $\mathcal{F}_{AND}$~\cite{rathee2021sirnn} takes as input boolean shares of $a$ and $b$ from $P_0$, $P_1$, and outputs boolean shares of $a\land b$.
The protocol for $\mathcal{F}_{AND}$~\cite{rathee2020cryptflow2, rathee2021sirnn, huang2022cheetah} can be implemented based on Beaver Triples, which can be implemented with $(1,16)$-$OT_2$, and communication complexity is $\lambda + 20$ bits.

\textbf{Equality Test \& Millionaires}. 
The functionality $\mathcal{F}_\text{eq}$~\cite{rathee2021sirnn} takes as input arithmetic shares of $a$ and $b$ from $P_0$, $P_1$, and returns shares of one bit $1\{a=b\}$.
The functionality $\mathcal{F}_\text{Mill}$~\cite{rathee2021sirnn} takes as input arithmetic shares of $a$ and $b$ from $P_0$, $P_1$, and returns shares of one bit $1\{a>=b\}$.
\cite{rathee2020cryptflow2} implement $\mathcal{F}_{eq}$ with the communication complexity below $\frac{3}{4}\lambda \ell + 9\ell$ and $\mathcal{F}_\text{Mill}$ below $\lambda \ell + 14\ell$.

\textbf{Share Conversion}.
The functionality $\mathcal{F}_\text{convert}^{N_1 \rightarrow N_2}$~\cite{rathee2021sirnn} takes as input arithmetic shares of $x$ on $\mathbb{Z}_{N_1}$ from $P_0$, $P_1$, and outputs arithmetic shares of $x$ on $\mathbb{Z}_{N_2}$.
The protocol for $\mathcal{F}_\text{convert}^{N_1 \rightarrow N_2}$ can be implemented based on $\mathcal{F}_\text{Mill}$ and $\mathcal{F}_\text{B2A}$~\cite{rathee2021sirnn}, and communication complexity is less than $\lambda \ell_1 + 14\ell_1 + \lambda + \ell_2$, $\ell_1=\lceil \log N_1\rceil, \ell_2 = \lceil \log N_2\rceil$.

\section{The Secure Counting Function}

\label{text: count_method}
In this section, we present our secure counting protocol, a core primitive underpinning many statistical analyses over private data such as the $\mathsf{Mode}$ and $\mathsf{CHISQ.Test}$. 
We formalize the secure counting functionality $\mathcal{F}_\text{PFC}$ in \autoref{func: pfc}, which securely constructs a frequency vector $\langle \vec{y}\rangle$ from a secret-shared vector $\langle \vec{x}\rangle$.
Beyond the general setting, we consider a practically important special case in which half of the input features are identical\footnote{This scenario commonly arises in statistical applications such as abnormal network traffic analysis~\cite{morsy2025vb, hennequin2024multi} and industrial equipment condition analysis~\cite{lahama2024implementation,kallitsis2016amon,brass2015local}.}, and introduce a customized CRT-based secure frequency counting protocol that partitions the feature domain into small sub-domains.

\myherebox{Functionality $\mathcal{F}_\text{PFC}^{m,n}$}{white!20}{white!10}{
  $\mathcal{F}_\text{PFC}$ interacts with $P_0$, $P_1$ and the simulator $\mathcal{S}$.

\underline{\textbf{Parameters:}} $n$ is the size of the dataset; $m$ is the size of feature domain. $\ell_n = \lceil \log n\rceil$, $\ell_m = \lceil \log m\rceil$.
  
\underline{\textbf{Input:}}
  \begin{itemize}[leftmargin=*]
  \item Upon receiving $(\mathsf{Input}, \mathsf{sid}, \langle \vec{x}\rangle_0)$ from $P_0$, record $\langle \vec{x}\rangle_0$ and send $(\mathsf{Input}, \mathsf{sid}, P_0)$ to $\mathcal{S}$, where $\langle \vec{x}\rangle_0 \in \mathbb{Z}_m^n$.
  \item Upon receiving $(\mathsf{Input}, \mathsf{sid}, \langle \vec{x}\rangle_1)$ from $P_1$, record $\langle \vec{x}\rangle_1$ and send $(\mathsf{Input}, \mathsf{sid}, P_1)$ to $\mathcal{S}$, where $\langle \vec{x}\rangle_1 \in \mathbb{Z}_m^n$.
  \end{itemize}
    
\underline{\textbf{Execution:}}
  \begin{itemize}[leftmargin=*]
   \item If both $\langle \vec{x}\rangle_0, \langle \vec{x}\rangle_1$ are recorded, reveal $\vec{x} := \langle \vec{x}\rangle_0 + \langle \vec{x}\rangle_1$, compute the frequency vector $\vec{y}$, uniform randomly sample vector $\langle \vec{y}\rangle_0 \gets \mathbb{Z}_{2^{\ell_n}}^m$, and compute $\langle \vec{y}\rangle_1 := \vec{y} - \langle \vec{y}\rangle_0$, $\langle \vec{y} \rangle_1 \in \mathbb{Z}_{2^{\ell_n}}^m$.
   \item Send $(\mathsf{Output},\mathsf{sid}, \langle \vec{y}\rangle_b)$ to $P_b$, $b \in \{0, 1\}$.
  \end{itemize}
}{The Ideal Functionality $\mathcal{F}_\text{PFC}$.}{func: pfc}

\subsection{Private Frequency Counting via Secure Shift}
\label{text: subpfs}
To get the private frequency vector, previous methods~\cite{zhang2011generic} get the frequency of $v$ through $n$ times equality test, i.e., $1\{x_i=v\}, i \in [n]$, convert the boolean shares to the arithmetic shares, and then sum up these equality test results.
The intuition of our private frequency counting is that if we combine the result of all $1\{x = v_j\}, j\in[m]$, we can get a boolean unit vector $\vec{e}^x$, where $\vec{e}^x$ equals $1$ only in the index $x$ and $0$ else.
To obtain the boolean vector $\vec{e}^x$, instead of performing $m$ times equality test, we can get it through a secure vector shift, as shown in \autoref{fig: pfs_shift}.
Given the secret shares of $\langle x\rangle$, we can get $\vec{e}$ through two vector circular shifts on the unit vector $\vec{e}^0$ with the offsets $\langle x\rangle_0, \langle x\rangle_1$, respectively.
\begin{figure}[htbp]
    \centering
    \includegraphics[scale=0.8]{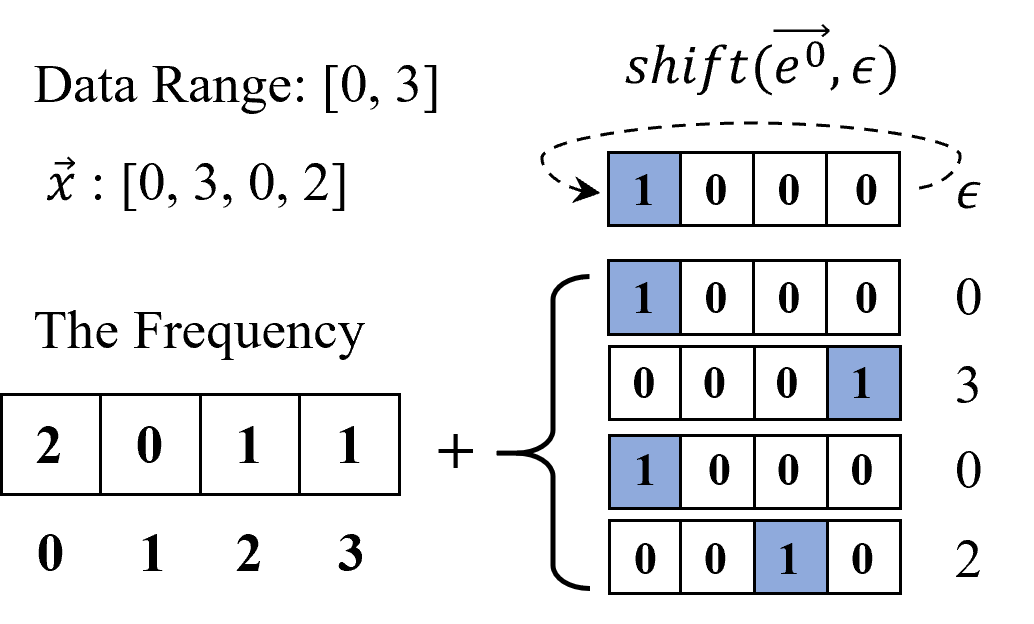}
    \caption{Illustration of shift-based secure frequency counting.}
    \Description{Diagram illustrating the shift-based private frequency counting approach.}
    \label{fig: pfs_shift}
\end{figure}

Specifically, these two circular shifts can be achieved in the following steps.
First $P_0$ can locally perform a circular shift with offset $\langle x\rangle_0$ and $\vec{e}^0$ to get the $\vec{e}^{\langle x\rangle_0}$.
To circular shift vector $\vec{e}^{\langle x\rangle_0}$ with offset $\langle x\rangle_1$, we leverage the functionality $\mathcal{F}_\text{RVOSE}$.
In the offline phase, parties invoke $\mathcal{F}_\text{RVOSE}$.
$P_0$ gets $\vec{a}, \vec{b}$, $P_1$ gets $\epsilon, \vec{c}$.
In the online phase, $P_0$ masks $\vec{e}^{\langle x\rangle_0}$ with $\vec{a}$, sends $\vec{d}:= \vec{e}^{\langle x\rangle_0} \oplus \vec{a}$ to $P_1$.
$P_1$ masks $\langle x\rangle_1$ with $\epsilon$, sends $\langle x\rangle_1 - \epsilon$ to $P_0$.
Receiving $\langle x\rangle_1 - \epsilon$, $P_0$ circular shifts $\vec{b}$ with offset $\langle x\rangle_1 - \epsilon$, outputs $\langle \vec{e}^{x}\rangle_0:= \text{shift}(\langle x\rangle_1 - \epsilon,\vec{b})$.
$P_1$ circular shifts $\vec{d}$ with $\epsilon$ offset, and computes $\vec{t}:=\text{shift}(\epsilon, \vec{d}) \oplus \vec{c}$, further circular shifts $\vec{t}$ with offset $\langle x\rangle_1 - \epsilon$, and outputs $\langle \vec{e}^{x}\rangle_1:= \text{shift}(\langle x\rangle_1 - \epsilon,\vec{t})$.
Performing the above steps for each input $x_i$, we can get the boolean share $\langle \vec{e}^{x_i} \rangle$.
Getting the boolean shares $\langle \vec{e}^{x_i} \rangle$, parties convert them into the arithmetic shares and locally sum up to get the secret shares of the frequency vector $\langle \vec{y}\rangle$.
The complete protocol is given in \autoref{pro: pfc}.

\myherebox{Protocol $\Pi_\text{PFC}^{m,n}$}{white!20}{white!10}{
   The protocol is parameterized by the dataset size $n$ and the feature domain $m$.
   $\ell_m = \lceil \log m\rceil, \ell_n = \lceil \log n\rceil$.

   \emph{$\mathsf{Input:}$} 
   For $b \in \{0, 1\}$, $P_b$ inputs the datasets $\langle \vec{x}\rangle_b \in \mathbb{Z}_m^n$. $m$ is the number of possible values.

   \emph{$\mathsf{Output:}$} 
   $P_b$ outputs the private frequency $\langle \vec{y}\rangle_b \in \mathbb{Z}_{2^{\ell_n}}^m$.

\underline{\textbf{Protocol:}}

\underline{\textbf{Offline:}}
\begin{enumerate}[label=\arabic*), leftmargin=*]
    \item For $i\in [n]$, $P_0$ samples vector $\vec{a}_i \gets \mathbb{Z}_{2}^m$ uniformly at random. 
    \item For $i\in [n]$, $P_0$ \& $P_1$ invoke $\mathcal{F}_\text{RVOSE}$ with $P_0$ inputs $\vec{a}_i$, and $P_0$ receives $\vec{b}_i \in \mathbb{Z}_2^m$, $P_1$ receives $\vec{c}_i \in \mathbb{Z}_2^m$ and $\epsilon_i \in \mathbb{Z}_{m}$.
\end{enumerate}
\underline{\textbf{Online:}}
\begin{enumerate}[label=\arabic*), leftmargin=*]
    \item For $i\in [n]$, $P_0$ initializes a boolean vector $\vec{e}_i \in \mathbb{Z}_{2}^m$, and sets $\vec{e}_i[j] := 1$ if $j = \langle \vec{x}[i]\rangle_0$ else $0$.
    \item For $i\in [n]$, $P_0$ computes $\vec{d}_i:=\vec{e}_i \oplus \vec{a}_i$ and sends it to $P_1$.
    \item For $i\in [n]$, $P_1$ computes $\langle \vec{x}[i]\rangle_1 - \epsilon_i$ and sends it to $P_0$.
    \item For $i\in [n]$, $P_1$ computes $\vec{t}_i := \text{shift}(\epsilon_i, \vec{d}_i)\oplus\vec{c}_i$, and computes $\text{shift}(\langle \vec{x}[i]\rangle_1 - \epsilon_i, \vec{t}_i)$ as $\langle \vec{y}''_i\rangle_1$.
    \item For $i\in [n]$, $P_0$ computes $\text{shift}(\langle \vec{x}[i]\rangle_1 - \epsilon_i, \vec{b}_i)$ as $\langle \vec{y}''_i\rangle_0$.
    \item For $i\in [n]$, $P_0$ \& $P_1$ invoke $\mathcal{F}_\text{B2A}$ with input $\langle \vec{y}''_i\rangle_0, \langle \vec{y}''_i\rangle_1$ and learns $\langle \vec{y}'_i\rangle_0, \langle \vec{y}'_i\rangle_1$.
    \item $P_b$ locally compute $\langle \vec{y}[j]\rangle_b := \sum_{i=0}^{n-1}\langle \vec{y}'_i[j]\rangle_b$.
\end{enumerate}
}{The Secure Frequency Counting Protocol.}{pro: pfc}

\textbf{Correctness.} The circular shift of a $m$ length vector $\vec{x}$ holds that if we have offset $\epsilon_0$, $\epsilon_1$ and $\epsilon=\epsilon_0 + \epsilon_1 \bmod m$, then the result of circular shift $\vec{x}$ with $\epsilon$ offset equals to twice circular shift of $\vec{x}$ with $\epsilon_0$ and $\epsilon_1$.
Given the correctness of $\mathcal{F}_\text{RVOSE}$, we have $\vec{t}\oplus\vec{b} = \text{shift}(\epsilon, \vec{e}^{\langle x\rangle_0} \oplus \vec{a})\oplus \vec{c} \oplus \vec{b} = \text{shift}(\epsilon, \vec{e}^{\langle x\rangle_0})$. The further circular shift with offset $\langle x\rangle_1 - \epsilon$ holds that $\langle \vec{y}''\rangle_0 \oplus \langle \vec{y}''\rangle_1 = \text{shift}(\langle x\rangle_1 - \epsilon, \vec{t}\oplus\vec{b})$. Such that the vector $\vec{e}^{0}$ shifts offset $\langle x\rangle_0 + \epsilon + \langle x\rangle_1 - \epsilon = x \bmod n$. 
Given the correctness of $\mathcal{F}_\text{B2A}$, the correctness of the frequency counting can easily follow.

\textbf{Efficiency.}
For each element in vector $\vec{x}$, the protocol $\Pi_\text{PFC}$ requires one call to $\mathcal{F}_\text{RVOSE}$ in the offline phase, $m$-length boolean vector and one index message communication, and $m$ calls to $\mathcal{F}_\text{B2A}$.
The offline phase requires $n3\lambda \log m$ bit communication with $2$ rounds.
The online phase requires $n(m+\ell_m + m\lambda + m\ell_n)$ bit communication with $3$ rounds.
The total cost is $n(3\lambda \log m + \ell_m + m(1 + \lambda + \ell_n))$ bits with $5$ communication rounds.
The total cost of previous frequency counting based on the equality test is $mn(\frac{3}{4}\lambda\ell_m + 9\ell_m + \lambda + \ell_n)$ with $\log \ell_m + 2$ communication rounds.
We also compare our protocol with the bit decomposition-based frequency counting design~\cite{wprfpeceny2025efficient}.
In this design, the total cost is $n(2\lambda l_m +40 l_m + m(\lambda + 20)/2 + m(\lambda + l_n))$ with $2l_m + 1$ rounds.
Asymptotically, the equality-based design requires $\mathcal{O}(mn(\log m + \log n))$. 
Our shift-based design instead requires $\mathcal{O}(n\log m + mn\log n)$.
The bit decomposition-based requires the same asymptotical complexity with our design, but has a higher communication round complexity.
Overall, the secure-shift design provides a lower complexity compared to the baselines.

\textbf{Security.}
We define the functionality $\mathcal{F}_\text{PFC}^{m,n}$ for secure frequency counting in \autoref{func: pfc} and describe the protocol in \autoref{pro: pfc}. We prove our protocol $\Pi_\text{PFC}^{m,n}$ realizes functionality $\mathcal{F}_\text{PFC}^{m,n}$.

\begin{theorem}
\label{theorem: pfc}
$\Pi_\text{PFC}^{m,n}$ UC realizes $\mathcal{F}_\text{PFC}^{m,n}$ in the $(\mathcal{F}_\text{B2A}, \mathcal{F}_\text{RVOSE})$ hybrid model against semi-honest probabilistic polynomial time (PPT) adversaries with statistical corruption.
\end{theorem}
\begin{proof}
CF. Appendix.~\ref{sec: pfc} for details.
\end{proof}

\subsection{Private Frequency Counting with CRT}
\label{text: subcrt}
In this subsection, we present a customized secure counting to get the mode of the input vector $\vec{x}$ when the mode’s frequency $f_\text{mode}$ exceeds $n/2$.
The intuition of our customized secure counting is that if the mode value $w$ has frequency $f_w > n/2$, then the mode $w_i$ of the dataset $\vec{x}_i$ satisfies $w_i = w \bmod t_i$, where $\vec{x}_i:= \{ d \bmod t_i \mid d \in \vec{x} \}$ and $t_i$ are co-prime. 
Following this insight, we break down the private mode on a large frequency vector into a series of short frequency sub-vectors by choosing $t_i$ according to the CRT, i.e., $t_i$ are co-prime, and $N=\prod_i t_i$ is slightly larger than $m$.
Then, obtaining the arithmetical share of the mode in each sub-vector $\langle w_i\rangle_b$, we can reconstruct $w$ with the CRT that: 
$w := \sum_{i=1}^{u} w_i \cdot M_i \cdot \beta_i \pmod{N}$, where
$M_i = \frac{N}{t_i}, \beta_i \equiv M_i^{-1} \bmod{t_i}$ are both public numbers, such that $w$ can be reconstruct locally.
\begin{figure}[htbp]
    \centering
    \includegraphics[scale=0.7]{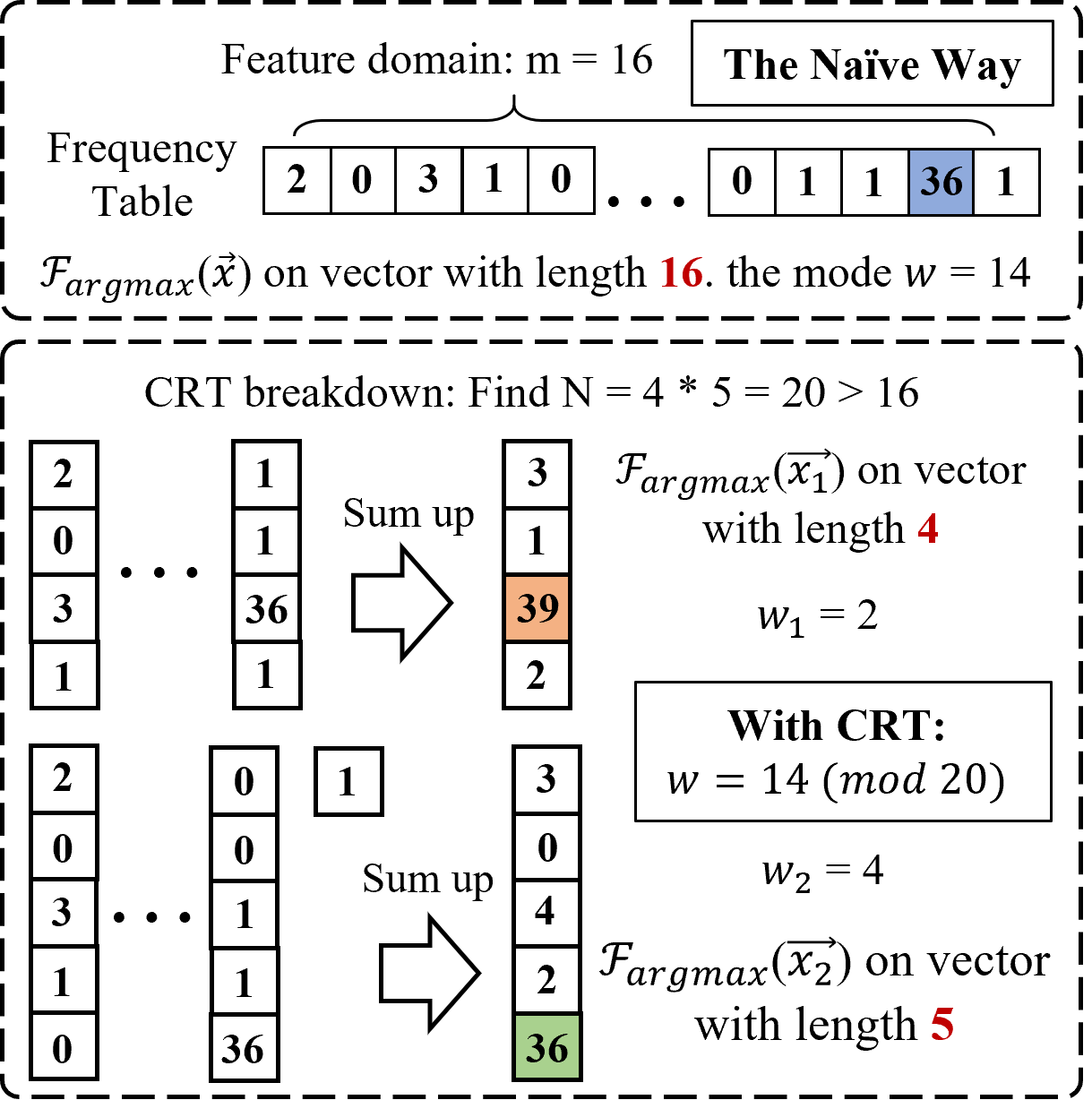}
    \caption{The mode function computed in a naive way and optimized with CRT. The data range is $[0, 15]$.}
    \label{fig: mode_crt}
    \Description[Short description]{}
\end{figure}

Specifically, we present our customized secure frequency counting via a toy example in \autoref{fig: mode_crt}.
Suppose the data range is $[0, 15]$ and the mode’s frequency exceeds half the dataset size.
Using the naive method, we need to perform $\mathcal{F}_\text{argmax}^{n}$~\cite{rathee2020cryptflow2} on the frequency vector with length $16$ to find the mode.
$\mathcal{F}_\text{argmax}^{n}$ gets secret shared vector $\langle \vec{x}\rangle_b$ as input and outputs the secret shared index of the max value in $\vec{x}$.
In contrast, with the CRT optimization, parties first convert the input $\langle \vec{x}\rangle \in \mathbb{Z}_{m}^n$ into $\langle \vec{x_i}\rangle \in \mathbb{Z}_{t_i}^{n}$.
Then parties compute the frequency vector of each $\langle \vec{x_i}\rangle$ with $\mathcal{F}_\text{PFC}^{t_i, n}$ to learns $\langle \vec{y_i}\rangle_b$, and perform $\mathcal{F}_\text{argmax}^{t_i}$ with input sub-domain frequency vectors to get the modes.
In our toy example, parties perform $\mathcal{F}_\text{argmax}^{4}$ and $\mathcal{F}_\text{argmax}^{5}$ on two sub-domain frequency vectors, respectively.
After getting the results of modes $\langle w'_i\rangle^{t_i}$ over the sub-domain frequency vectors, parties convert $\langle w'_i\rangle^{t_i}$ into $\langle w'_i\rangle^{N}$, and compute the mode $\langle w\rangle^N$ via the CRT.
The complete algorithm is in \autoref{algo: mode}.

\myherebox{Protocol $\Pi_\text{mode}^{m,n,\mathbb{T}}$}{white!20}{white!10}{
   The protocol is parameterized by the dataset size $n$ and the feature domain $m$.
   $\ell_m = \lceil \log m\rceil, \ell_n = \lceil \log n\rceil$.
   $t_i$ are $u$ co-prime numbers in set $\mathbb{T}$, and $N=\prod_i t_i \ge m$.

   \emph{$\mathsf{Input:}$} 
   For $b \in \{0, 1\}$, $P_b$ inputs the datasets $\langle\vec{x}\rangle_b \in \mathbb{Z}_{m}^{n}$.

   \emph{$\mathsf{Output:}$} 
   $P_b$ outputs the share $\langle w\rangle_b^N$, $w = \mathsf{mode}(\vec{x})$.

\underline{\textbf{Protocol:}}
\begin{enumerate}[label=\arabic*), leftmargin=*]
    \item For $i\in[u]$, $P_0$ \& $P_1$ invoke $\mathcal{F}_\text{convert}^{m \rightarrow t_i}$ with input $\langle \vec{x}\rangle_0, \langle \vec{x}\rangle_1$ and learns $\langle \vec{x_i}\rangle_0^{t_i}, \langle \vec{x_i}\rangle_1^{t_i}$.
    \item For $i\in[u]$, $P_0$ \& $P_1$ invoke $\mathcal{F}_\text{PFC}^{t_i, n}$ with input $\langle \vec{x_i}\rangle_0^{t_i}, \langle \vec{x_i}\rangle_1^{t_i}$ and learn $\langle \vec{y_i}\rangle_0 \in \mathbb{Z}_{2^{\ell_n}}^{t_i}, \langle \vec{y_i}\rangle_1 \in \mathbb{Z}_{2^{\ell_n}}^{t_i}$.
    \item For $i\in[u]$, $P_0$ \& $P_1$ invoke $\mathcal{F}_\text{argmax}^{t_i}$ with input $\langle \vec{y_i}\rangle_0, \langle \vec{y_i}\rangle_1$ and learn secret shares of the index of the max value $\langle w'_i\rangle_0^{t_i}, \langle w'_i\rangle_1^{t_i}$.
    \item For $i\in[u]$, $P_0$ \& $P_1$ invoke $\mathcal{F}_\text{convert}^{t_i\rightarrow N}$ with input $\langle w'_i\rangle_0^{t_i}, \langle w'_i\rangle_1^{t_i}$ and learns $\langle w_i\rangle_0^N, \langle w_i\rangle_1^N$.
    \item $P_b$ locally reconstructs the arithmetical share $\langle w\rangle_b$ mode based on the CRT.
\end{enumerate}
}{The Customized Secure $\mathsf{mode}$.}{algo: mode}

\textbf{Correctness.}
The correctness that $\langle w_i\rangle^N$ are the mode can be easily followed via the correctness of $\mathcal{F}_\text{convert}^{N_1 \rightarrow N_2}$, $\mathcal{F}_\text{PFC}^{m,n}$, and $\mathcal{F}_\text{argmax}^{n}$.
Given the correctness of $\langle w_i\rangle^N$, $w = \langle w\rangle_0 + \langle w\rangle_1 = \sum_{i=1}^{u} (\langle w_i\rangle_0 + \langle w_i\rangle_1) \cdot M_i \cdot \beta_i = \sum_{i=1}^{u} w_i \cdot M_i \cdot \beta_i \pmod{N}$, which is correct with the CRT.

\textbf{Efficiency.} 
$\Pi_\text{mode}^{m,n,\mathbb{T}}$ requires $u$ times $\mathcal{F}_\text{convert}^{m \rightarrow t_i}$, $\mathcal{F}_\text{PFC}^{t_i,n}$, $\mathcal{F}_\text{argmax}^{t_i}$ and $\mathcal{F}_\text{convert}^{t_i \rightarrow N}$.
The $\mathcal{F}_\text{convert}^{m \rightarrow t_i}$ requires $(\lambda + 14)u\ell_m + u\lambda + \sum_{i=0}^{u-1}\ell_{t_i}$ bits with $\log \ell_m + 2$ communication rounds.
$\mathcal{F}_\text{PFC}^{t_i,n}$ requires $n\sum_{i=0}^{u-1}(3\lambda \log t_i + \ell_{t_i} + t_i(1+\lambda + \ell_n))$ with $5$ communication rounds.
$\mathcal{F}_\text{argmax}^{t_i}$ requires $\sum_{i=0}^{u-1}(\lambda \ell_n + 16 \ell_n+4\lambda + 2\ell_{t_i})(t_i-1)$ with $2\log \ell_n\log \hat{t}$ communication rounds, where $\hat{t}$ is the max $t_i$.
$\mathcal{F}_\text{convert}^{t_i \rightarrow N}$ requires $(\lambda + 14)\sum_{i=0}^{u-1}\ell_{t_i} + u\lambda + u\ell_N$ bits with $\log \ell_{\hat{t}} + 2$ communication rounds.
Using the naive way, the communication complexity to computes mode of the $m$ length frequency vector is $n(3\lambda\log m + \ell_m + m(1+\lambda+\ell_n))$ for $n$ times $\mathcal{F}_\text{PFC}^{m,n}$  and $(\lambda \ell_n + 16\ell_n + 4\lambda + 2\ell_m)(m-1)$ for $\mathcal{F}_\text{argmax}^{m}$.
The communication rounds is $2\log \ell_n \log m + 5$.
Asymptotically, the CRT optimization communication complexity is about $\mathcal{O}(n\log m + \sum{t_i} n\log n + \sum{t_i} \log m)$.
The naive way is about $\mathcal{O}(n\log m + mn\log n + m \log m)$.

\textbf{Security}. We provide the security proof of $\Pi_\text{mode}^{m,n,\mathbb{T}}$ in Appendix.\ref{sec: mode}.

\section{The Secure Sort Function}

\label{text: order_method}
In this section, we present our secure counting sort protocol, which supports statistical functions (e.g., $\mathsf{Rank}$) on large datasets but small feature domains.
We formalize the secure sort functionality $\mathcal{F}_\text{sort}$ in \autoref{func: sort}, which takes a secret-shared vector $\langle \vec{x}\rangle$ as input and outputs a secret-shared sorted vector $\langle \vec{y}\rangle$.
In addition, to efficiently support statistical functions such as $\mathsf{Percentiles}$, we reformulate them as interval-testing problems based on the frequency vector, thereby avoiding the need to fully sort the input vector.

\myherebox{Functionality $\mathcal{F}_\text{sort}^{m,n}$}{white!20}{white!10}{
  $\mathcal{F}_\text{sort}$ interacts with $P_0$, $P_1$ and the simulator $\mathcal{S}$.

\underline{\textbf{Parameters:}} $n$ is the size of the dataset; $m$ is the size of feature domain. $\ell_m=\lceil \log m\rceil$, $\ell_n=\lceil \log n\rceil$.
  
\underline{\textbf{Input:}}
  \begin{itemize}[leftmargin=*]
  \item Upon receiving $(\mathsf{Input}, \mathsf{sid}, \langle \vec{x}\rangle_0)$ from $P_0$, record $\langle \vec{x}\rangle_0$ and send $(\mathsf{Input}, \mathsf{sid}, P_0)$ to $\mathcal{S}$, where $\langle \vec{x}\rangle_0 \in \mathbb{Z}_m^n$.
  \item Upon receiving $(\mathsf{Input}, \mathsf{sid}, \langle \vec{x}\rangle_1)$ from $P_1$, record $\langle \vec{x}\rangle_1$ and send $(\mathsf{Input}, \mathsf{sid}, P_1)$ to $\mathcal{S}$, where $\langle \vec{x}\rangle_1 \in \mathbb{Z}_m^n$.
  \end{itemize}
    
\underline{\textbf{Execution:}}
  \begin{itemize}[leftmargin=*]
   \item If both $\langle \vec{x}\rangle_0, \langle \vec{x}\rangle_1$ are recorded, reveal $\vec{x} := \langle \vec{x}\rangle_0 + \langle \vec{x}\rangle_1$, compute the ascending sorted vector $\vec{y}$, uniformly random sample $\langle\vec{y}\rangle_0 \gets \mathbb{Z}_{2^{\ell_m}}^n$, and compute $\langle \vec{y}\rangle_1 := \vec{y} - \langle\vec{y}\rangle_0$.
   \item Send $(\mathsf{Output},\mathsf{sid}, \langle \vec{y}\rangle_b)$ to $P_b$, $b \in \{0, 1\}$.
  \end{itemize}
}{The Ideal Functionality $\mathcal{F}_\text{sort}$.}{func: sort}

\subsection{The Secure Counting Sort}
\label{text: subsort}
The secure counting sort consists of two main steps: 1)the frequency vector construction and 2)the element insertion.
The first step can be achieved through $\mathcal{F}_\text{PFC}^{m,n}$.
We here present the element insertion using our secure segment-indicator vector.
The secure segment-indicator vector is an arithmetic secret share of a vector with elements in index $a$ to $b$ is $1$ else $0$, abbr, $\vec{1}_{a:b}$ (the index $a,b$ here starts from $1$, $\vec{1}_{a:b}$ denotes a zero vector, when $a > b$).
Given the secret share of $\vec{1}_{a:b}$, we can insert value $v$ via a scale multiplication $v\cdot \vec{1}_{a:b}$. 
We compute the vector $\vec{1}_{a:b}$ through $\vec{1}_{1:b} - \vec{1}_{1:a-1}$
To get vector $\vec{1}_{1:b}$ with input $b$, we can derive it from the unit vector $\vec{e}^b$ via $\vec{1}_{1:b}[i] := \sum_{j=i}^n\vec{e}^b[j]$.
$\vec{e}^b$ can be constructed similarly with a secure shift paradigm.

\begin{figure}[htbp]
    \centering
    \includegraphics[scale=1.0]{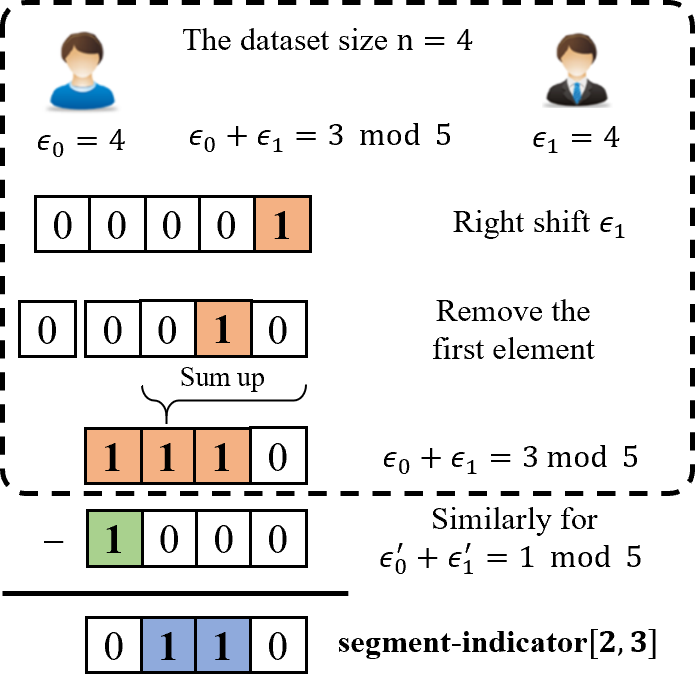}
    \caption{Example of constructing $\vec{1}_{2:3}$ for data range $[0, 3]$.}
    \label{fig: block_indicator}
    \Description[Short description]{}
\end{figure}

Specifically, after obtaining the private frequency vector, parties first compute the prefix sum of the private frequency vector $\vec{q}$, where $\vec{q}[i]$ is the end index of value $i$ in the final sorted vector and $\vec{q}[i]+1$ is the start index of value $i+1$.
In other words, the segment-indicator vector for value $i$ is $\vec{1}_{\vec{q}[i-1]+1:\vec{q}[i]}$. 
As the frequency ranges from $0$ to $n$, it needs to convert the private frequency vector $\vec{q}$ into shares over $\mathbb{Z}_{n+1}$ to ensure the correctness. 
We present the concrete steps to get $\langle \vec{1}_{a:b}\rangle$ with a toy example, assuming $n=4, a=2, b=3$, as shown in \autoref{fig: block_indicator}.
Given $\epsilon_0 + \epsilon_1 = 3 \bmod 5$, parties can compute the $\langle \vec{e}^3\rangle$ leveraging $\mathcal{F}_\text{RVOSE}$.
Then, parties remove the first element of $\langle \vec{e}^3\rangle$ such that for frequency $0$, the parties can obtain secret shares of a zero vector.
With $\langle \vec{\hat{e}}^3\rangle:=\langle \vec{e}^3[1:]\rangle$ parties can compute $\langle \vec{1}_{1:3} \rangle$ through $\langle \vec{1}_{1:3}[i] \rangle := \sum_{j=i}^n\langle \vec{\hat{e}}^3[j]\rangle$.
Similarly, parties can get $\langle \vec{1}_{1:1} \rangle$ with input $\epsilon_0', \epsilon_1'$, where $\epsilon_0' + \epsilon_1' = 1 \bmod 5$.
With $\langle \vec{1}_{1:3} \rangle, \langle \vec{1}_{1:1} \rangle$ parties can locally compute $\langle \vec{1}_{2:3} \rangle := \langle \vec{1}_{1:3} \rangle - \langle \vec{1}_{1:1} \rangle$.
Using $\langle \vec{y}_{i}'\rangle$ to denote the secret share of the segment-indicator vector of value $i$, we can get the final sorted vector through $\langle \vec{y}\rangle:=\sum_{i=0}^{m-1}\langle \vec{y}_{i}'\rangle\cdot (i+1) - \vec{1}$.
The detailed steps are shown in \autoref{pro: csort}.

\textbf{Correctness}. 
By the correctness of $\mathcal{F}_\text{PFC}^{m,n}$ and $\mathcal{F}_\text{RVOSE}$, we get the secret shared frequency vector $\langle \vec{q}' \rangle$ of the input $\langle \vec{x}\rangle$.
In the sorted vector, the value $i$ is inserted into the vector behind all values less than $i$ such that the segment-indicator vector is $\vec{1}_{a:b}$, where $a=1+\sum_{j=0}^{i-1} \vec{q}'[j]$, $b = a+\vec{q}'[i]$.
Then, $\langle \vec{1}_{a:b}\rangle = \langle \vec{1}_{1:b}\rangle - \langle \vec{1}_{1:a-1}\rangle$, and $\langle \vec{1}_{1:a}[i] \rangle = \sum_{j=i}^n\langle \vec{\hat{e}}^a[j]\rangle$.
$\langle \vec{\hat{e}}^a\rangle=\langle \vec{e}^a[1:]\rangle$.
The first element of $\vec{e}^a$ is removed to ensure the value with frequency $0$ outputs a zero segment-indicator vector.
The correctness of $\langle \vec{e}^a\rangle$ follows similar to $\Pi_\text{PFC}^{m,n}$.

\textbf{Efficiency}. $\Pi_\text{CSort}^{m,n}$ first invokes $\mathcal{F}_\text{PFC}^{m,n}$, which costs $n3\lambda\log m$ bits with $2$ rounds in the offline phase, and $n(m+\ell_m+m\lambda+m\ell_n)$ bits with $3$ rounds in the online phase. Then invokes the share conversion to translate the shares in $\mathbb{Z}_{n+1}$, which costs less than $\lambda\ell_n + 15\ell_n + \lambda$ with $\log \ell_m + 1$ rounds.
Then, $m-1$  times circular shift using the arithmetic version of $\mathcal{F}_\text{RVOSE}$, which costs $(m-1)3\lambda\log (n+1)$ bits with $2$ rounds in the offline phase and $(m-1)(\ell_n+(n+1)\ell_m)$ bits with $1$ rounds in the online phase.
The total communication complexity is about $\mathcal{O}(mn\log (mn))$, and the round complexity is $\mathcal{O}(\log \log m)$.
The communication and round complexity of the QuickSort~\cite{hou2023ciphergpt} are  
$O(n\log n(\log (mn))$ and $O(\log n\log\log m)$ respectively.
The communication and round complexities of RadixSort~\cite{wprfpeceny2025efficient} are
$\mathcal{O}(mn\log m + m^{1/t}n\log n)$ and $\mathcal{O}(\log m)$, respectively,
where $t$ depends on the selected radix.

\textbf{Security}.
We define the functionality $\mathcal{F}_\text{sort}^{m,n}$ for secure sort in \autoref{func: sort} and describe the protocol in \autoref{pro: csort}. 
We prove our protocol $\Pi_\text{CSort}^{m,n}$ realizes functionality $\mathcal{F}_\text{sort}^{m,n}$.

\begin{theorem}
\label{theorem: csort}
$\Pi_\text{CSort}^{m,n}$ UC realizes $\mathcal{F}_\text{sort}^{m,n}$ in the $(\mathcal{F}_\text{PFC}^{m,n}, \mathcal{F}_\text{RVOSE})$ hybrid model against semi-honest probabilistic polynomial time (PPT) adversaries with statistical corruption.
\end{theorem}
\begin{proof}
CF. Appendix.~\ref{sec: csort} for details.
\end{proof}

\myherebox{Protocol $\Pi_\text{CSort}^{m,n}$}{white!20}{white!10}{
   The protocol is parameterized by the dataset size $n$ and the feature domain $m$.
   $\ell_m = \lceil \log m\rceil, \ell_n = \lceil \log n\rceil$.
   
   \emph{$\mathsf{Input:}$} 
   For $b \in \{0, 1\}$, $P_b$ inputs vector $\langle \vec{x}\rangle_b \in \mathbb{Z}_m^n$.

   \emph{$\mathsf{Output:}$} 
   $P_b$ outputs the sorted vector $\langle \vec{y}\rangle_b \in \mathbb{Z}_{2^{\ell_m}}^n$.

\underline{\textbf{Protocol:}}

\underline{\textbf{Offline:}}
\begin{enumerate}[label=\arabic*), leftmargin=*]
    \item For $i\in [m\!-\!1]$, $P_0$ samples vector $\vec{a}_i \gets \mathbb{Z}_{2^{\ell_m}}^{n+1}$ uniformly at random. 
    \item For $i\in [m\!-\!1]$, $P_0$ \& $P_1$ invoke $\mathcal{F}_\text{RVOSE}$ with $P_0$ inputs $\vec{a}_i$, and $P_0$ receives $\vec{b}_i \in \mathbb{Z}_{2^{\ell_m}}^{n+1}$, $P_1$ receives $\vec{c}_i \in \mathbb{Z}_{2^{\ell_m}}^{n+1}$ and $\epsilon_i \in \mathbb{Z}_{n+1}$.
\end{enumerate}

\underline{\textbf{Protocol:}}
\begin{enumerate}[label=\arabic*), leftmargin=*]
    \item $P_0$ \& $P_1$ invoke $\mathcal{F}_\text{PFC}^{m,n}$ with input $\langle \vec{x}\rangle_0, \langle \vec{x}\rangle_1$ and learn $\langle \vec{q}''\rangle_0, \langle \vec{q}''\rangle_1$.
    \item $P_0$ \& $P_1$ invoke $\mathcal{F}_\text{convert}^{2^{\ell_n}\rightarrow n+1}$ with input $\langle \vec{q}''\rangle^{2^{\ell_n}}_0, \langle \vec{q}''\rangle^{2^{\ell_n}}_1$ and learn $\langle \vec{q}'\rangle_0^{n+1}, \langle \vec{q}'\rangle_1^{n+1}$.
    \item For $i \in [m\!-\!1]$, $P_b$ computes $\langle \vec{q}[i]\rangle_b^{n+1} := \sum_{j=0}^i\langle \vec{q}'[j]\rangle_b^{n+1}$.
    \item For $i \in [m\!-\!1]$, $P_0$ initializes $\vec{v}_i \in \mathbb{Z}_{2^{\ell_m}}^{n + 1}$, and sets $\vec{v}_i[j] := 1$ if $j=\langle \vec{q}[i]\rangle_0$ else $0$.
    \item For $i \in [m\!-\!1]$, $P_0$ sends $\vec{d}_i :=\vec{v}_i -\vec{a}_i$ to $P_1$.
    \item For $i \in [m\!-\!1]$, $P_1$ sends $\langle \vec{q}[i] \rangle_1 - \epsilon_i$ to $P_0$.
    \item For $i \in [m\!-\!1]$, $P_1$ computes $\vec{t}_i := \text{shift}(\epsilon_i, \vec{d}_i) + \vec{c}_i$, and sets $\text{shift}(\langle \vec{q}[i]\rangle_1-\epsilon_i, \vec{t}_i)$ as $\langle \vec{y}'''\rangle_1$.
    \item For $i \in [m\!-\!1]$, $P_0$ sets $\text{shift}(\langle \vec{q}[i]\rangle_1 - \epsilon_i, \vec{b}_i)$ as $\langle \vec{y}'''_i\rangle_0$.
    \item For $i \in [m\!-\!1]$, $P_b$ computes $\langle \vec{y}''_i[j]\rangle_b:=\sum_{u=j}^{n} \langle\vec{y}'''_i[u]\rangle_b, j \in [1, n]$.
    \item $P_0$ sets $\langle \vec{y}''_{m-1} \rangle_0$ to an all-ones vector, while $P_1$ sets $\langle \vec{y}''_{m-1} \rangle_1$ to the zero vector.
    \item For $i \in [1, m\!-\!1]$, $P_b$ compute $\langle \vec{y}_i'\rangle:=\langle \vec{y}_{i}''\rangle - \langle \vec{y}_{i-1}''\rangle$ and $\langle \vec{y}_0'\rangle:=\langle \vec{y}_{0}''\rangle$.
    \item $P_b$ computes $\langle \vec{y}\rangle:=\sum_{i=0}^{m-1}\langle \vec{y}_{i}'\rangle\cdot (i+1) - \vec{1}$.
\end{enumerate}
}{Protocol of $\Pi_\text{CSort}$}{pro: csort}

\subsection{The Secure Counting Percentile}
\label{text: subkth} 
$\mathsf{Percentile}$ outputs the $k$-th smallest value from the dataset that represents the point below which a given percentage of data falls, where the given percentage determines $k$.
As a special case of sorting, $\mathsf{Percentile}$ does not need to sort the entire input vector.
Such that instead of inserting all the elements according to the frequency counting vector, we construct $\Pi_\text{percentile}^{m,n,k}$ based on the following observation.
For the $\mathsf{Percentile}$ operation, we note that the $k$-th smallest value $i$ denotes that $k$ falls in the $i$-th interval of the prefix-sum frequency vector.
Following the above idea of the interval test, we evaluate the k-th smallest value by judging which interval $k$ falls in.

Specifically, parties first construct the private prefix-sum frequency vector $\langle \vec{z}\rangle$ via $\mathcal{F}_\text{PFC}^{m,n}$. 
Then, to judge whether $k$ falls in the interval $(\vec{z}[i], \vec{z}[i+1]]$, parties compute the indicator $1\{k \le \vec{z}[i+1]\}$ by invoking $\mathcal{F}_\text{Mill}$ on inputs $k$ and $\langle \vec{z}[i+1]\rangle$, obtaining a private bit $\langle t[i+1]\rangle$.
Next, parties obtain $1\{k > \vec{z}[i]\}$ by setting $\langle \vec{u'}[i]\rangle_b := \langle \vec{t}[i]\rangle_b \oplus b$, and securely combine $\langle \vec{u'}[i]\rangle$ and $\langle \vec{t}[i+1]\rangle$ using $\mathcal{F}_\text{AND}$ to derive $\langle u[i]\rangle$, which indicates whether $k$ lies in the interval $(\vec{z}[i], \vec{z}[i+1]]$.
From these values, parties construct a private unit vector with one position corresponding to the target interval.
They then invoke $\mathcal{F}_\text{B2A}$ to convert this boolean vector into an arithmetic unit vector $\langle \vec{w}\rangle$, and finally compute the output percentile value as $\langle \omega \rangle := \sum_{i=0}^{m-1} i \cdot \langle w[i]\rangle$, which yields the $k$-th smallest element in the domain.
The complete protocol is presented in \autoref{algo: exkth}.

\textbf{Correctness}. 
The correctness of $\Pi^{m,n,k}_{\text{percentile}}$ follows from the correctness that the unit vector $\vec{v}$ indicates which interval $k$ falls in.
Specifically, the correctness of $\vec{v}$ is established as follows.
For the first entry, we have
$\vec{v}[0] = \vec{t}[0] = 1\{k \le \vec{z}[0]\}$.
For each $i \in [1, m-2]$, it holds that
$\vec{v}[i] = \vec{t}[i] \land \vec{u'}[i-1]$ = $1\{k \le \vec{z}[i]\} \land 1\{k > \vec{z}[i-1]\} = 1\{\vec{z}[i-1] < k \le \vec{z}[i]\}$.
Finally, for the last entry,
$\vec{v}[m-1] = \vec{u'}[m-2] = 1\{k>\vec{z}[m-2]\}$.
Therefore, $\vec{v}$ is a unit vector that correctly encodes the interval containing $k$.

\textbf{Efficiency}. $\Pi_\text{kth}$ invokes $m-1$ times $\mathcal{F}_\text{Mill}$ to get the result of $1\{k \le \vec{z}[i]\}, i \in [0, m-2]$.
Then need $m-2$ times $\mathcal{F}_\text{AND}$ to get the interval test result of $1\{\vec{z}[i] < k \le \vec{z}[i+1]\}, i \in [1, m-2]$.
Finally $m$ times $\mathcal{F}_\text{B2A}$ to get the result of $k$. 
This part of the interval test communication complexity is 
$m(\lambda(\ell_n + 2) + 14\ell_n+\ell_m+20) - \lambda(\ell_n+2)-14\ell_n-40$. 
This part of the interval test communication round is $\log \log m + 4$ rounds.
Adding the cost of secure frequency vector construction, the total communication complexity is about $\mathcal{O}(m\log m + n\log m+ mn\log n)$, and the round complexity is about $O(\log\log m)$.
The communication complexity of the SOTA~\cite{hou2023ciphergpt, hamada2012practically} is $\mathcal{O}(n \log nm)$ and a communication round complexity $O(\log n\log\log m)$.

\textbf{Security}.
We provide the security proof of $\Pi_\text{percentile}^{m,n,k}$ in Appendix.\ref{sec: kth}.

\myherebox{Protocol $\Pi_\text{percentile}^{m,n,k}$}{white!20}{white!10}{
   The protocol is parameterized by the dataset size $n$ and the feature domain $m$.
   $k$ denotes the index of the percentile value in the sorted vector.
   $\ell_m = \lceil \log m\rceil, \ell_n = \lceil \log n\rceil$.
   
   \emph{$\mathsf{Input:}$} 
   For $b \in \{0, 1\}$ $P_b$ inputs the private prefix-sum frequency $\langle \vec{z}\rangle_b \in \mathbb{Z}_{2^{\ell_n}}^m$ and the private $\langle k\rangle_b$.

   \emph{$\mathsf{Output:}$} 
   $P_b$ outputs the $k$-th smallest value $\langle \omega\rangle_b$. 

\underline{\textbf{Protocol:}}
\begin{enumerate}[label=\arabic*), leftmargin=*]
    \item $P_0$ \& $P_1$ invoke $\mathcal{F}_\text{PFC}^{m,n}$ with input $\langle \vec{x}\rangle_0, \langle \vec{x}\rangle_1$ and learns $\langle \vec{z}'\rangle_0, \langle \vec{z}'\rangle_1$.
    \item $P_b$ computes $\langle \vec{z}[i]\rangle_b:=\sum_{j=0}^i\langle \vec{z'}[j]\rangle_b$.
    \item $P_0$ \& $P_1$ invoke $\mathcal{F}_\text{Mill}$ with input $k$ and $\langle \vec{z}[i] \rangle_0, \langle \vec{z}[i] \rangle_1, i \in [0, m\!-\!2]$ and gets $\langle t_i\rangle_0, \langle t_i\rangle_1$, $t_i = 1\{k \le \vec{z}[i]\}$.
    \item $P_b$ sets $\langle \vec{u}'[i] \rangle_b := \langle \vec{t}[i]\rangle_b \oplus b, i \in [0, m\!-\!2]$.
    \item $P_0$ \& $P_1$ invoke $\mathcal{F}_\text{AND}$ with input $\langle \vec{t}[i\!+\! 1]\rangle_0, \langle \vec{t}[i\!+\! 1]\rangle_1$, $\langle \vec{u}'[i]\rangle_0, \langle \vec{u}'[i]\rangle_1$, and learns $\langle \vec{u}[i]\rangle_0, \langle \vec{u}[i]\rangle_1, i \in [0, m\!-\!3]$.
    \item $P_b$ constructs boolean vector $\langle \vec{v}\rangle$, where $\langle \vec{v}[0]\rangle_b:=\langle \vec{t}[0]\rangle_b$, $\langle \vec{v}[i\!+\!1]\rangle_b := \langle \vec{u}[i]\rangle_b, i\in[0,m\!-\!3]$, and $\langle \vec{v}[m\!-\!1]\rangle_b:=\langle \vec{u}'[m\!-\!2]\rangle_b$.
    \item $P_0$ \& $P_1$ invoke $\mathcal{F}_\text{B2A}$ with input $\langle \vec{v}[i]\rangle_0, \langle \vec{v}[i]\rangle_1, i \in [0, m\!-\!1]$ and learns $\langle \vec{w}[i]\rangle^{2^{\ell_m}}_0, \langle \vec{w}[i]\rangle^{2^{\ell_m}}_1, i \in [0, m\!-\!1]$.
    \item $P_b$ locally computes $\langle \omega\rangle_b := \sum_{i=0} ^ {m-1} i\cdot \langle \vec{w}[i]\rangle_b$.
\end{enumerate}
}{The Secure Percentiles Protocol}{algo: exkth}

For the max or min operation, it is special as its order index falls into the interval endpoints, which enables us to replace the $\mathcal{F}_\text{Mill}$ with $\mathcal{F}_\text{eq}$.
In contrast to the general interval test whose endpoints are typically different numbers, here the interval endpoints can be the same (the case appears when there are no elements whose value is $x_i$, such that $z_{i-1} = z_{i}$).
Such that there will be more than one interval endpoint hit to indicate the maximum value.
After performing $\mathcal{F}_\text{eq}(\langle z_i\rangle_b, n), i \in [m]$, we will get a boolean secret share vector $\langle \vec{u}\rangle_b^B$, where $\vec{u}[j] = 1$ if $j \ge \mathsf{Max}(\vec{x})$ else $0$.
To get the result of a unit vector indicates the max value, $P_b$ can be computed as follows: $\vec{u}'[0]:= \vec{u}[0], \vec{u}'[i]:= \vec{u}[i-1] \oplus \vec{u}[i]$.
Obtaining $\vec{u}'$, we can get the max value through $m$ times $\mathcal{F}_\text{B2A}$ and a local sum up.
The extended communication complexity besides $\mathcal{F}_\text{PFC}^{m,n}$ is $m(\frac{3}{4}\lambda + 9 \ell_n + \lambda + \ell_m)$.
We present the complete protocol $\Pi_\text{CMax}$ in \autoref{appendix: ex_max}.

\section{The Secure Non-linear Math Functions}

\label{text: math_method}
In this section, we present our bisection-based secure MSNZ bit protocol, which constitutes the dominant computational cost in secure non-linear math functions.
We adopt the functionality $\mathcal{F}_\text{MSNZB}$ defined in~\cite{rathee2021sirnn}, which takes a secret-shared value $\langle x\rangle$ as input and outputs a secret-shared vector $\langle \vec{z}\rangle^B$ s.t. $\vec{z}[i] = 1$ if $2^i \le x < 2^{i+1}$ and 0 otherwise.
Furthermore, we formally characterize the relationship between the degree of an approximation polynomial and its valid approximation interval, and propose a secure interval-compression protocol that offers an orthogonal efficiency optimization for secure non-linear math functions.

\subsection{The Optimized Most Significant Non-Zero Bit Protocol}
\label{text: subdeflation}
To get the MSNZ bit of $\langle x \rangle$, Sirnn~\cite{rathee2021sirnn} realizes first invoking the $\mathcal{F}_{\text{DigDec}}$ to convert the arithmetical secret shares of $x$ into the arithmetical secret shares of $x_i \in \mathbb{Z}_{2^c}, i\in [d], c =  \ell/d$, where $x = x_0 || \dots||x_{d-1}$.
Then getting the $k_i = \text{MSNZB}(x_i), v_i =1$ if $x_i = 0$ else $0$ using private lookup table (LUT), 
and then combining the result to get $\text{MSNZB}(x) := \sum_i k_i \cdot (1 \oplus v_i)\cdot \prod_{j>i}v_j$, which requires $\mathcal{O}(d)$ communication rounds to locate the MSNZ $x_i$.
The intuition of our $\Pi_\text{MSNZB}$ is that the MSNZ $x_i$ locates in the non-zero half part of $x$.
Such that we can obliviously locate the MSNZ $x_i$ based on the bisection idea recursively.

Specifically, parties first decompose $\langle x\rangle$ to $\langle x_0\rangle||\langle x_1\rangle$, where $x=x_0||x_1$ $x_0$ and $x_1$ have the same bit length through the invocation of $\mathcal{F}_\text{Mill}$ and $\mathcal{F}_\text{B2A}$.
Then parties check whether $x_0$ is zero through the invocation of $\mathcal{F}_\text{eq}$.
Leveraging the result $1\{x_0 = 0\}$, parties can obliviously choose $x_b$ as $x'$, where $b=1 - 1\{x_0=0\}$.
Then we apply the same steps on $x'$ to get a shorter message, which can finally be efficiently evaluated by the private LUT.
To get the final result, assume that the MSNZB of $x'$ is $k$, $\text{MSNZB}(x)=k+\ell/2\cdot(1\oplus b)$.
To get the boolean secret share of vector $\langle \vec{z}\rangle^B$ with input $\langle \text{MSNZB}(x)\rangle$, we can construct it via a secure vector circular shift similar in $\Pi_\text{PFC}$.
For input length $\ell$ that is not a power of $2$, we can perform zero extension to convert the input to $\mathbb{Z}_{2^t}$ to ensure that during each step the input can be decomposed into two with the same bit length.
The functionality $\mathcal{F}_\text{LUT}$ and $\mathcal{F}_\text{one-hot}$~\cite{rathee2021sirnn} can be implemented with $(1,k)$-$\text{OT}$.
We provide the complete protocol $\Pi_\text{MSNZB}$ in \autoref{pro: msnzb}.

\textbf{Correctness}.
For $\ell$ bit input $x=x_0||x_1$, we have $\text{MSNZB}(x) = \text{MSNZB}(x_0) + \ell/2$ if $x_0 \neq 0$ else $\text{MSNZB}(x) = \text{MSNZB}(x_1)$.
When the input to $\text{MSNZB}$ is small enough, the correctness is ensured through the correctness of $\mathcal{F}_\text{LUT}$.

\textbf{Efficiency}.
For $\ell$ bit input $x$, our method consists of four parts to get the MSNZ bit of $x$.
The first step is to obliviously find the most significant block, where each step requires one invocation of $\mathcal{F}_\text{Mill}$, $\mathcal{F}_\text{B2A}$, $\mathcal{F}_\text{eq}$ and $\mathcal{F}_\text{MUX}$, whose communication complexity is less than $(\frac{7}{4}\nu +3d)\lambda + 26\nu$, where $\ell=dc$, $\nu=\sum_{i=1}^{\log d}(\frac{\ell}{2^i})$ with $2\nu + 3\log d$ communication rounds.
The second step is the evaluation of LUT, whose communication complexity is $2\lambda + 2^c\iota$ with $2$ communication rounds.
The third step is the result combination, which requires $\log d$ times $\mathcal{F}_\text{B2A}$, whose communication complexity is $\log d(\lambda+\iota)$ with $2\log d$ communication rounds.
The fourth step is the one-hot conversion, whose communication complexity is $ 2\lambda+ \ell^2$ with $2$ communication rounds.
As for the concrete parameter $\ell=32, c=8$ suggested in Sirnn~\cite{rathee2021sirnn}, our protocol has $1.47\times$ less communication complexity.

\textbf{Security}.
We describe the protocol in \autoref{pro: msnzb}. 
We prove our protocol $\Pi_\text{MSNZB}$ realizes functionality $\mathcal{F}_\text{MSNZB}$ in Sirnn~\cite{rathee2021sirnn}.

\begin{theorem}
\label{theorem: msnzb}
$\Pi_\text{MSNZB}$ UC realizes $\mathcal{F}_\text{MSNZB}$ in the $\mathcal{F}_\text{Mill}$, $\mathcal{F}_\text{B2A}$, $\mathcal{F}_\text{eq}$, $\mathcal{F}_\text{MUX}$, $\mathcal{F}_\text{one-hot}$, $\mathcal{F}_\text{LUT}$ hybrid model against semi-honest probabilistic polynomial time (PPT) adversaries with statistical corruption.
\end{theorem}
\begin{proof}
CF. Appendix.~\ref{sec: msnzb} for details.
\end{proof}

\myherebox{Protocol $\Pi_\text{MSNZB}$}{white!20}{white!10}{
   
   \emph{$\mathsf{Input:}$} 
   For $b \in \{0, 1\}$ $P_b$ inputs $\langle x\rangle_b \in \mathbb{Z}_{2^\ell}$.

   \emph{$\mathsf{Output:}$} 
   $P_b$ outputs $\{\langle z_i\rangle_b^B\}_{i\in[\ell]}$ s.t. $z_i = 1$ if $2^i \le x < 2^{i+1}$ and 0 otherwise.

\underline{\textbf{Protocol:}}
\begin{enumerate}[label=\arabic*), leftmargin=*]
    \item $P_b$ sets $\langle x\rangle_b$ as $\langle x^1\rangle_b$
    \item For $j =1 $ to $\log d$:
    \begin{enumerate}[label=\arabic*), leftmargin=*]
        \item $P_b$ locally decompose $\langle x^j\rangle_b^{2^{\ell_j}}$ into $\langle u^j\rangle_b^{2^{\ell_{j+1}}}||\langle v^j\rangle_b^{2^{\ell_{j+1}}}$, where $\ell_{j+1}=\ell_j/2$.
        \item $P_b$ invokes $\mathcal{F}_\text{Mill}$ with inputs $\langle v^j\rangle_b$ and learns $\langle w^j\rangle_b^B$, where $w^j = 1\{\langle v^j\rangle_0+\langle v^j\rangle_1 \ge 2^{\ell_{j+1}}\}$.
        \item $P_b$ invokes $\mathcal{F}_\text{B2A}$ with input $\langle w^j\rangle_b^B$ and learns $\langle t^j \rangle_b$.
        \item $P_b$ invokes $\mathcal{F}_{\text{eq}}$ with input $\langle t^j\rangle_b, \langle u^j\rangle_b$ and learns $\langle \sigma^j \rangle_b^B$ , where $\sigma^j = 1\{0=\langle t^j\rangle_0 + \langle u^j\rangle_0+\langle t^j\rangle_1 + \langle u^j\rangle_1\}$.
        \item $P_b$ invokes $\mathcal{F}_{\text{MUX}}$ with input $\langle \sigma^j\rangle_b^B, \langle t^j \rangle_b + \langle u^j\rangle_b, \langle v^j\rangle_b$ and learns $\langle x^{j+1}\rangle_b$.
    \end{enumerate}

    \item $P_b$ invokes $\mathcal{F}_\text{LUT}$ with inputs $\langle x^{d+1}\rangle_b$ and learns $\langle \widetilde{z}^{d}\rangle_b \in \mathbb{Z}_{2^{\iota}}$, where $\iota = \log \ell$.
    \item For $j =1 $ to $\log d$:
    \begin{enumerate}[label=\arabic*), leftmargin=*]
        \item $P_b$ invokes $\mathcal{F}_{\text{B2A}}$ with input $\langle \sigma^{d-j}\rangle_b^B$ and learns $\langle h'\rangle_b^{2^\iota}$, then computes $\langle h\rangle_b:= \langle h'\rangle_b \cdot\frac{\ell}{2^{d-j}}$.
        \item $P_b$ computes $\langle \widetilde{z}^{d-j}\rangle_b:=\langle \widetilde{z}^{d-j+1}\rangle_b + \langle h\rangle_b$.
    \end{enumerate}
    \item $P_b$ invokes $\mathcal{F}_{\text{one-hot}}(\langle \widetilde{z}^{1}\rangle_b)$ and learns $\{\langle z_i\rangle_b^B\}_{i\in[\ell]}$.
\end{enumerate}
}{The Secure MSNZB Protocol}{pro: msnzb}

\subsection{The Secure Interval Compression Protocol}
\label{text: sublog_approx}
In this subsection, we present an alternative efficiency optimization direction for secure non-linear math functions through the secure interval compression protocol.
As discussed in \autoref{text: subdeflation}, $x$ can be deflated into the interval $[2^k, 2^k(1+\epsilon)]$, where $k \in \mathbb{Z}$.
Leveraging the $\mathcal{F}_\text{MSNZB}$, the $\epsilon$ is a constant equals $1$.
In the following, we present the efficency improvement brought from the interval compression when applying the polynomial approximation or the iterative methods like Newton's iteration.
In the following, we take $1/x$ as an example.

As for the polynomial approximation, to minimize the max approximation error with Chebyshev polynomial approximation, we have the following lemma.
\begin{lemma}\label{lemma: approx_error}
    The minimax approximation error of $1/x$ on interval $[2^k, 2^k(1+\epsilon)]$ is less than $\frac{\epsilon^{n+1}}{2^{2n+1}2^k}$.
\end{lemma}
\begin{proof}
    According to the minimax approximation error in~\cite{powell1981approximation}, we have that for the non-linear function $f(x)$ the n-degree approximation polynomial $p_n(x)$ on the range $[a, b]$ has the error
    \[
    \begin{aligned}
        \max_{a\le x \le b}|f(x) - p_n(x)|
        &\le \frac{(b-a)^{n+1}}{2^{2n+1}(n+1)!} \\
        &\quad \cdot \max_{a \le x \le b}|f^{(n+1)}(x)|.
    \end{aligned}
    \]
    For $f(x) = 1 / x$ on the interval $[2^k, 2^k(1+\epsilon)]$, we have $$Err \le \frac{\epsilon^{n+1}}{2^{2n+1}2^k}$$.
    which implies Lemma~\ref{lemma: approx_error}.
\end{proof}
Leveraging Lemma~\ref{lemma: approx_error}, we can see that if we only use $\mathcal{F}_\text{MSNZB}$, the error equals to $\frac{1}{2^{2n+1}2^k}$.
Although we can normalize $x$ into the interval $[2^k, 2^{k+1}]$ with larger $k$ to increase the convergence, it needs a longer bit to represent the input to avoid the overflow during the polynomial computation, which increases the overhead.
Our alternative direction is to minimize the $\epsilon$, i.e., compress the interval then $\epsilon=1/2$.
Such that the approximation error is $\frac{1}{2^{3n+2}2^k}$, which gives a faster convergence.
That is, given the same precision target, we can also use lower polynomial degree to bring an efficiency improvement.
As for the iterative method, the approximation error relates to the initial value.
In SIRNN~\cite{rathee2021sirnn}, the initial values are assigned with a private LUT.
Following our interval compression idea, we can reduce about half of the LUT size. 

To compress the interval such that $[2^k, 2^k(1+\epsilon)]$, where $\epsilon=1/2$, we propose the following lemma.
\begin{lemma}\label{lemma: range reduce}
    For any $x \in [a+b,a+2b)$, $a,b \in \mathbb{R}, a > 0,b > 0$, one can find $c \in \mathbb{R}$ and a positive number $x'$ such that
    $$ x' = cx, x' \in [a, a+b) $$
\end{lemma}
\begin{proof}
    One can write as $x'=cx$.
    such that $a+b\le cx < a+2b$. 
    To ensure this inequality holds, we can get the following two inequalities
    \[
        \left\{
        \begin{aligned}
        c\cdot(a+2b) &< a+b \\
        c \cdot(a+b) &\ge a
        \end{aligned}
        \right.
    \]
    We work out the $\frac{a}{a+b}\le c < \frac{a+b}{a+2b}$.
    We can compute that
    \begin{align}
        \frac{a+b}{a+2b} - \frac{a}{a+b} &= \frac{(a+b)^2 - a(a+2b)}{(a+2b)(a+b)} \nonumber\\
        &= \frac{b^2}{(a+2b)(a+b)} > 0 \nonumber
    \end{align}
    Thus we find $c\!=\!\frac{a}{a+b}$,
    which implies Lemma~\ref{lemma: range reduce}.
\end{proof}
The proof of Lemma~\ref{lemma: range reduce} also presents the concrete method to normalize $x$ in interval $[2^k, 2^k\cdot2)$ to $[2^k, 2^k\cdot 1.5)$.
Specifically, take $k=0$, $x\in[1,2)$ as an example, parties input $\langle x\rangle_b$, $1.5$ to $\mathcal{F}_\text{Mill}$.
After the execution, parties get the output $\langle \eta\rangle_b^B$, $\eta = 1\{x>1.5\}$ from $\mathcal{F}_\text{Mill}$.  
According to Lemma~\ref{lemma: range reduce}, for $a=1,b=1/2$, we can find $c=2/3$. 
The final $ x':= x + \eta (cx-x)$, which requires one $\mathcal{F}_\text{MUX}$ with input $\langle \eta\rangle^B$ and $\langle cx-x\rangle^{2^\ell}$.

\section{Experiments}

\label{text: experiment}
\textbf{Experimental setup}. Our experiments were conducted on a single machine with Intel(R) Core(TM) i5-14400F CPU @ 4.70GHz and 32.0GB RAM running Ubuntu 22.04.
To simulate various network conditions, we used Linux Traffic Control (tc), emulating a LAN environment with 10 Gbps bandwidth and 0.125 ms round-trip (RTT) latency and a WAN environment with 1 Gbps bandwidth and 50 ms RTT latency.
All experiments were built upon the EzPC library\footnote{https://github.com/mpc-msri/EzPC}, with the underlying IKNP-style OT protocols.
We use 32-out-of-16 to represent the decimals.
In the following, we first present our benchmark of the three functions: 1) secure frequency counting, 2) secure sorting, and 3) the non-linear math function. Then we give the result of the case study.
Our implementations are publicly available at https://anonymous.4open.science/r/PSC-FC2D/.

\subsection{Secure counting function} 
\label{text: experi_count}
We evaluate the communication cost and running time of secure frequency counting under two scalability settings: varying the data size $n$ with a fixed domain size $m=1024$, and varying the domain size $m$ with a fixed data size $n=4096$.
As baselines, we compare \sysname with three representative approaches: equality-based counting instantiated with state-of-the-art secure equality tests~\cite{rathee2021sirnn, rathee2020cryptflow2}, DPF-based\footnote{\url{https://github.com/xingpz2008/dealerless-FSS_public}} counting based on distributed function secret sharing~\cite{xing2025distributedFSS}, and bit-decomposition-based counting~\cite{wprfpeceny2025efficient}.

The results are reported in \autoref{fig: count_comm} and \autoref{tab: count_time}. 
Overall, \sysname consistently outperforms all baselines in both communication and running time.
In terms of communication, \sysname reduces the communication cost by $3.1$-$8.6\times$ compared with equality-based counting, $4.7$-$21.0\times$ compared with DPF-based counting, and $1.4$-$1.9\times$ compared with bit-decomposition-based counting.
For running time, \sysname achieves the best performance under both LAN and WAN settings. 
Under the LAN setting, \sysname achieves a $4.0$-$7.9\times$ speedup over equality-based counting and a $1.1$-$3.0\times$ speedup over bit-decomposition-based counting. 
Under the WAN setting, where network latency becomes more significant, \sysname further achieves a $1.3$-$7.2\times$ speedup over equality-based counting, a $1.6$-$2.3\times$ speedup over bit-decomposition-based counting.
For the DPF-based counting baseline, the correlated key generation under 2PC requires securely evaluating a large number of PRG expansions and correction-word computations, which makes it substantially slower than the other methods in the 2PC setting.
These results demonstrate that \sysname scales efficiently with both the data size and the domain size.

\begin{table}[!ht]
\centering
\caption{Running-time overhead comparison for counting schemes under LAN and WAN settings. DPF, BitDec, and Equality are short for DPF-based, bit-decomposition-based, and Equality-based schemes, respectively. When varying $n$, we fix $m=1024$; when varying $m$, we fix $n=4096$.
}
\label{tab: count_time}
\footnotesize
\begin{tblr}{
  width = \columnwidth,
  colspec = {@{}Q[c]Q[c]Q[c]*{4}{X[c]}@{}},
  cells = {c},
  rowsep = 1.5pt,
  colsep = 2pt,
  cell{1}{1} = {r=2}{},
  cell{1}{2} = {r=2}{},
  cell{1}{3} = {r=2}{},
  cell{1}{4} = {c=4}{},
  cell{3}{1} = {r=8}{},
  cell{3}{2} = {r=4}{},
  cell{7}{2} = {r=4}{},
  cell{11}{1} = {r=8}{},
  cell{11}{2} = {r=4}{},
  cell{15}{2} = {r=4}{},
  hline{1,3,19} = {-}{0.08em},
  hline{2} = {4-7}{0.08em},
  hline{7,11,15} = {-}{0.05em},
}
Network & Varying & Value & Runtime Overhead (s) &        &            &          \\
        &         &       & \makecell[c]{DPF} & \makecell[c]{BitDec} & \makecell[c]{Equality} & \makecell[c]{Ours} \\
LAN     & $m$     & 16    & 193.77    & 0.03       & 0.04       & \textbf{0.01}     \\
        &         & 64    & 277.22    & 0.08       & 0.20       & \textbf{0.05}     \\
        &         & 256   & 345.14    & 0.29       & 1.29       & \textbf{0.19}     \\
        &         & 1024  & 445.15    & 1.05       & 6.09       & \textbf{0.77}     \\     
        & $n$     & 2048  & 223.47    & 0.53       & 3.15       & \textbf{0.46}     \\
        &         & 4096  & 442.95    & 1.06       & 6.12       & \textbf{0.82}     \\
        &         & 8192  & 918.81    & 2.05       & 12.96      & \textbf{1.79}     \\
        &         & 16384 & 2143.02   & 4.42       & 25.28      & \textbf{3.89}     \\
        
WAN     & $m$     & 16    & 15125.56  & 1.12       & 0.92       & \textbf{0.70}     \\
        &         & 64    & 21707.92  & 1.89       & 2.03       & \textbf{0.92}     \\
        &         & 256   & 28112.28  & 3.27       & 7.77       & \textbf{1.42}     \\
        &         & 1024  & 34533.01  & 8.45       & 29.36      & \textbf{4.34}     \\
        & $n$     & 2048  & 17262.74  & 5.18       & 15.46      & \textbf{2.41}     \\
        &         & 4096  & 34517.30  & 8.38       & 29.99      & \textbf{4.34}     \\
        &         & 8192  & 69050.15  & 14.74      & 57.36      & \textbf{8.28}     \\
        &         & 16384 & 138098.30 & 27.34      & 116.00     & \textbf{16.14}  
        
\end{tblr}
\end{table}

\begin{figure}[htbp] 
    \centering    
    \begin{minipage}[c]{0.49\linewidth}
        \centering
        \includegraphics[width=\textwidth]{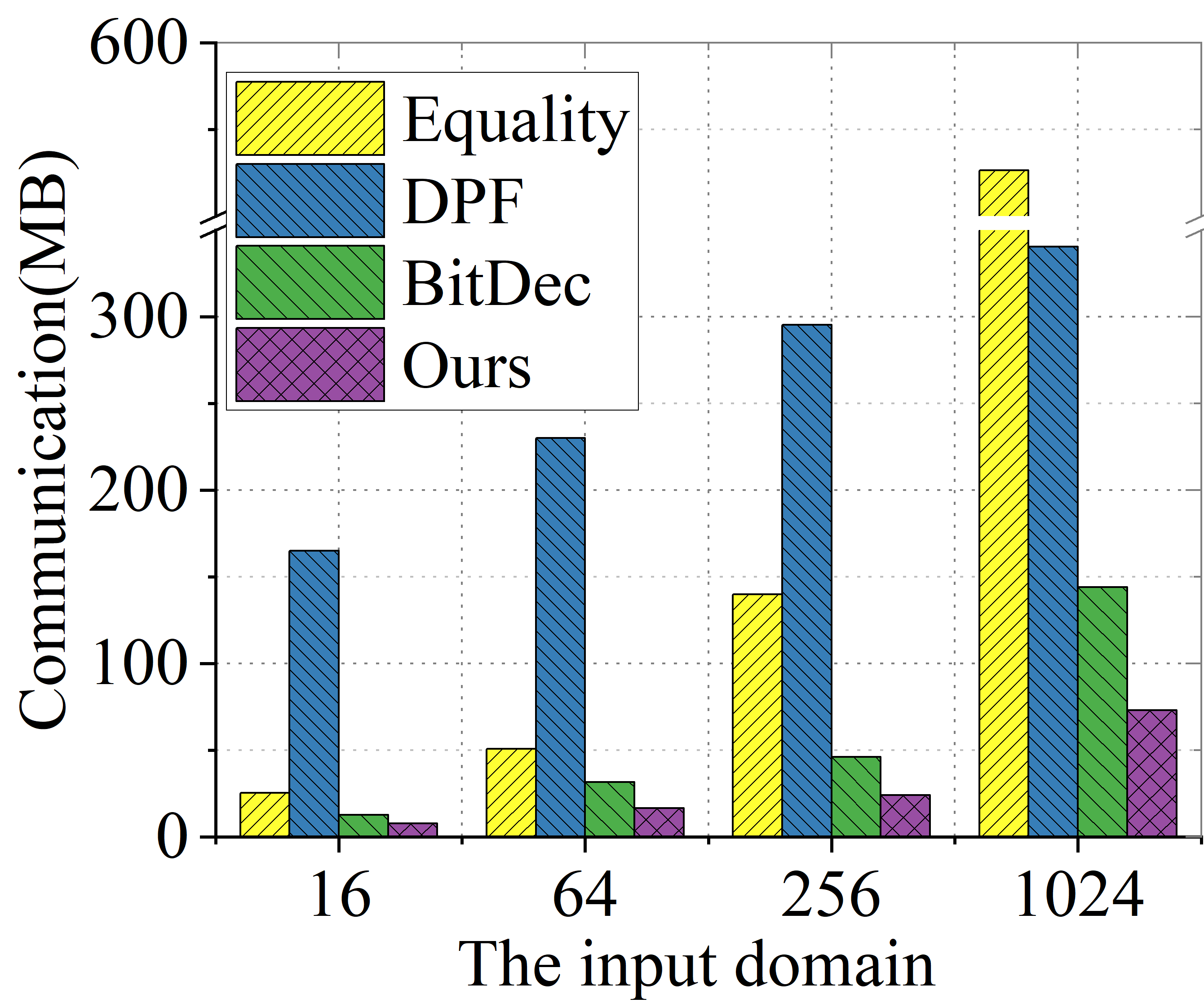}
    \end{minipage}
    \begin{minipage}[c]{0.49\linewidth}
        \centering
        \includegraphics[width=\textwidth]{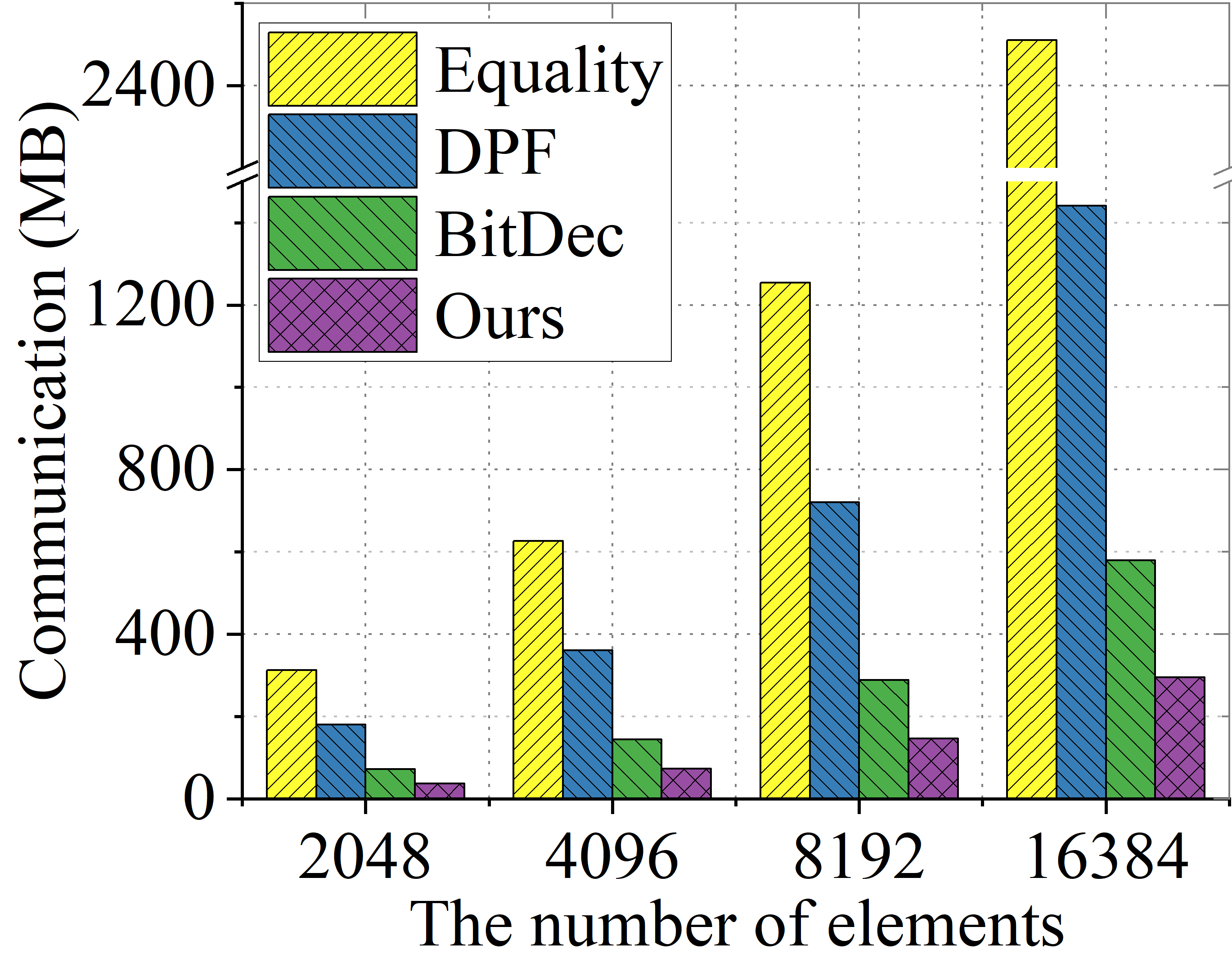}
    \end{minipage}

    \caption{Communication overhead of secure frequency counting under varying data sizes and domain sizes.}
    \label{fig: count_comm}
    \Description[Short description]{}
\end{figure}

\subsection{Secure order function}
Similar to the secure frequency counting setting, we evaluate the communication cost and running time of secure sorting under two scalability settings: varying the data size $n$ with a fixed domain size $m=1024$, and varying the domain size $m$ with a fixed data size $n=4096$.
As baselines, we compare \sysname with three representative approaches: shuffle-based QuickSort instantiated with state-of-the-art secure comparison~\cite{rathee2020cryptflow2, rathee2021sirnn} and secure shuffle~\cite{chase2020secret}, DCF-based sorting based on distributed function secret sharing~\cite{xing2025distributedFSS}, and RadixSort from~\cite{wprfpeceny2025efficient}.
The results are reported in \autoref{tab: sort_time} and \autoref{fig: sort_comm}.
The DCF-based sorting is omitted from \autoref{fig: sort_comm} for readability, since its communication cost is at the GB scale ($>$ 50GB).

When varying the data size $n$ with a fixed domain size $m=1024$, \sysname incurs approximately $2\times$ higher communication cost than the QuickSort and RadixSort baselines.
Despite this increase in communication volume, \sysname requires significantly fewer communication rounds, $\mathcal{O}(\log\log m)$ compared to $\mathcal{O}(\log n \log\log m)$ for the QuickSort baseline and $\mathcal{O}(\log m)$ for the RadixSort baseline, which results in lower overall running time.
Specifically, under the LAN setting, \sysname achieves a $4.7$-$5.9\times$ speedup over the QuickSort baseline across all evaluated values of $n$.
Under the WAN setting, where network latency dominates, the advantage becomes substantially more pronounced: \sysname reduces the running time by roughly $90\times$ over QuickSort baseline and $4.7\times$ over RadixSort baseline.

When varying the domain size $m$ with a fixed data size $n=4096$, the communication cost and running time of QuickSort stay almost unchanged as $m$ increases, since its cost primarily scales with the dataset size $n$. 
In contrast, RadixSort exhibits a moderate growth trend because its cost is affected by the number of radix decomposition rounds, which increases with $m$.
For \sysname, both communication cost and running time grow more noticeably as $m$ increases, due to its domain-size-dependent computation pattern.
Specifically,
under the LAN setting, \sysname is $21.4\times$, $10.1\times$, and $3.4\times$ faster than RadixSort when $m=16$, $64$, and $256$, respectively.
However, when $m$ increases to $1024$, the running time of \sysname becomes comparable to RadixSort, with a slight slowdown of $1.1\times$.
In terms of communication, \sysname also requires less communication than RadixSort when $m$ is small, while its communication cost becomes about $2.1\times$ higher than RadixSort at $m=1024$.

Overall, these results indicate that \sysname is particularly effective in optimizing running time and communication cost when the feature domain size $m < \log^2 n$.

\begin{table}[!ht]
\centering
\caption{Running-time overhead comparison for sorting schemes under LAN and WAN settings. DCF is short for DCF-based sorting. When varying $n$, we fix $m=1024$; when varying $m$, we fix $n=4096$.}
\label{tab: sort_time}
\begin{tblr}{
  width = \columnwidth,
  colspec = {@{}Q[c]Q[c]Q[c]*{4}{X[c]}@{}},
  cells = {c},
  rowsep = 1.5pt,
  colsep = 2pt,
  cell{1}{1} = {r=2}{},
  cell{1}{2} = {r=2}{},
  cell{1}{3} = {r=2}{},
  cell{1}{4} = {c=4}{},
  cell{3}{1} = {r=8}{},
  cell{3}{2} = {r=4}{},
  cell{7}{2} = {r=4}{},
  cell{11}{1} = {r=8}{},
  cell{11}{2} = {r=4}{},
  cell{15}{2} = {r=4}{},
  hline{1,3,19} = {-}{0.08em},
  hline{2} = {4-7}{0.08em},
  hline{7,11,15} = {-}{0.05em},
}
Network & Varying & Value & Runtime Overhead (s) &       &       &      \\
        &         &       & \makecell[c]{DCF} & \makecell[c]{QuickSort} & \makecell[c]{RadixSort} & \makecell[c]{Ours} \\
LAN     & $m$     & 16    & 34610.61  & 14.66   & 1.07   & \textbf{0.05}  \\
        &         & 64    & 47301.48  & 14.71   & 1.62   & \textbf{0.16}  \\
        &         & 256   & 52696.67  & 14.70   & 2.16   & \textbf{0.63}  \\
        &         & 1024  & 61931.15  & 14.75   & \textbf{2.71}   & 2.87  \\     
        & $n$     & 2048  & 27208.71  & 5.58    & 1.28   & \textbf{1.18}  \\
        &         & 4096  & 67119.64  & 14.64   & 2.76   & \textbf{2.53}  \\
        &         & 8192  & 153826.28 & 42.58   & 7.92   & \textbf{7.54}  \\
        &         & 16384 & 337420.98 & 136.03  & \textbf{15.29}  & 22.97 \\
WAN     & $m$     & 16    & $3.49\mathrm{e}6$ & 1113.81 & 13.48  & \textbf{1.23}  \\
        &         & 64    & $4.55\mathrm{e}6$ & 1104.57 & 20.41  & \textbf{1.54}  \\
        &         & 256   & $5.76\mathrm{e}6$ & 1114.85 & 27.22  & \textbf{3.12}  \\
        &         & 1024  & $6.98\mathrm{e}6$ & 1109.80 & 34.65  & \textbf{9.77}  \\
        & $n$     & 2048  & $2.91\mathrm{e}6$ & 557.76  & 20.96  & \textbf{4.94}  \\
        &         & 4096  & $7.65\mathrm{e}6$ & 1119.99 & 35.08  & \textbf{9.54}  \\
        &         & 8192  & $1.62\mathrm{e}7$ & 2232.24 & 95.82  & \textbf{20.19} \\
        &         & 16384 & $3.80\mathrm{e}7$ & 4507.88 & 186.52 & \textbf{46.82} 
        
\end{tblr}
\end{table}

\begin{figure}[htbp] 
    \centering    
    \begin{minipage}[c]{0.49\linewidth}
        \centering
        \includegraphics[width=\textwidth]{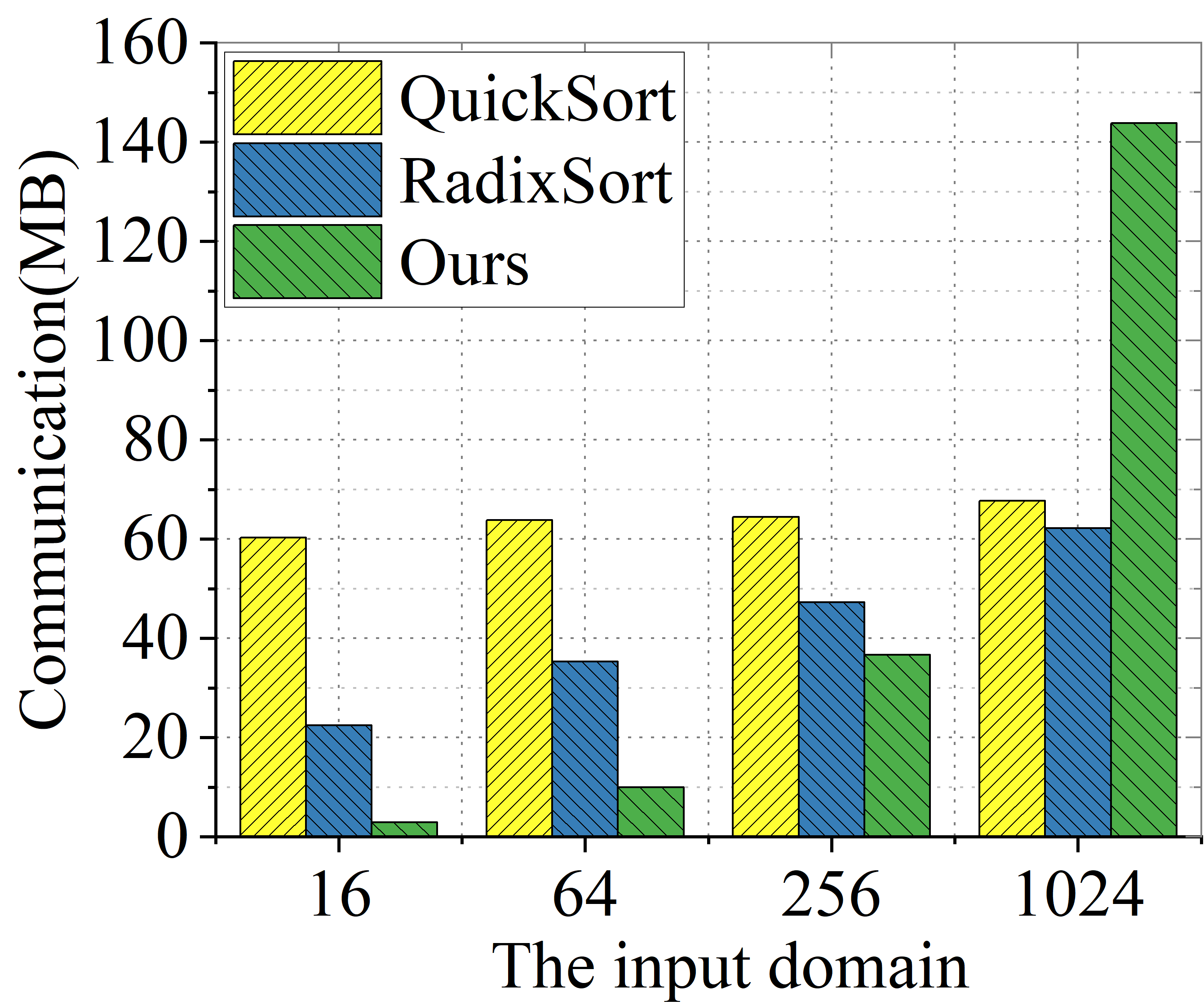}
    \end{minipage}
    \begin{minipage}[c]{0.49\linewidth}
        \centering
        \includegraphics[width=\textwidth]{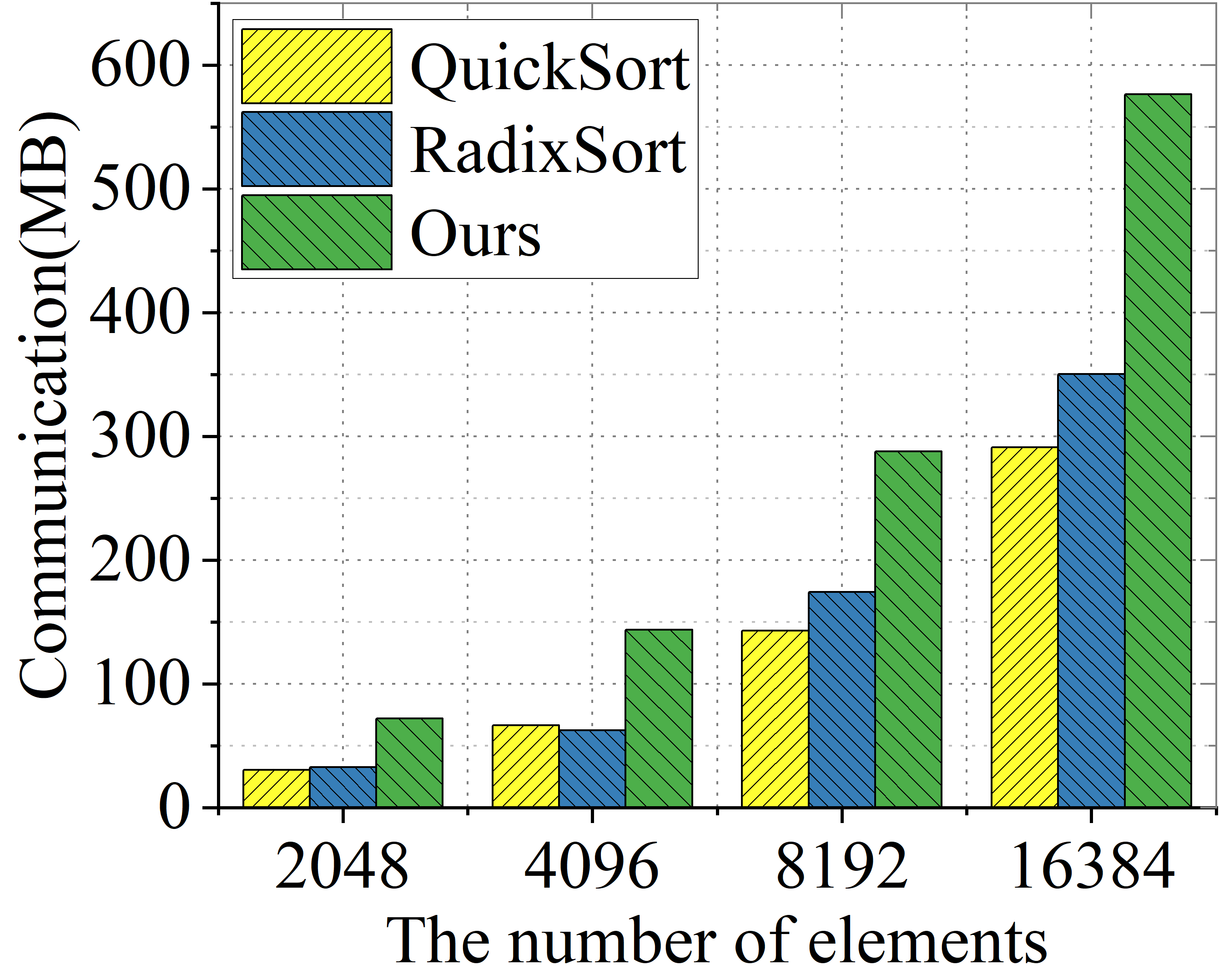}
    \end{minipage}
    \caption{Communication overhead of secure sorting under varying data sizes and domain sizes,excluding the GB-scale DCF-based sorting for readability.}
    \label{fig: sort_comm}
    \Description{A diagram showing Communication comparison.}
\end{figure}

\subsection{Secure non-linear math function}
In this subsection, we evaluate our protocols for secure non-linear math functions, including the tree-based $\Pi_{\text{MSNZB}}$ and its application to functions $1/x$, $1/\sqrt{x}$, and $\ln(x)$.

To evaluate the efficiency improvement of our secure MSNZB protocol, we test $10^5$ numbers under both LAN and WAN settings and set the MSNZB in Sirnn~\cite{rathee2021sirnn} as the baseline.
As shown in \autoref{tab: msnzb_cmpwithsirnn}, our tree-based MSNZB achieves consistent improvements in both communication and runtime.
Specifically, it reduces communication by $1.47\times$ and achieves $2.48\times$ and $1.61\times$ faster execution under LAN and WAN settings, respectively.
These gains stem from the hierarchical tree structure, which requires only a single invocation of the LUTs, leading to a more communication-efficient protocol.
Since MSNZB serves as a fundamental building block for most nonlinear functions, this optimization directly accelerates the overall evaluation of such functions.

\begin{table}[!ht]
\centering
\caption{Performance comparison of MSNZB implementations.}
\label{tab: msnzb_cmpwithsirnn}
\begin{tblr}{
  cells = {c},
  cell{1}{1} = {r=2}{},
  cell{1}{2} = {c=2}{},
  cell{1}{4} = {r=2}{},
  vline{2-4} = {3-4}{},
  hline{1,3,5} = {-}{0.08em},
  hline{2} = {2-3}{0.08em},
  hline{4} = {-}{0.05em},
}
Method                                  & Total Time (in sec)     &                          & {Comm./ Inst.\\ (in KB)} \\
                                        & LAN                     & WAN                      &                          \\
\sysname                 & 1.49                    & 8.32                     & 1.14                     \\
\cite{rathee2021sirnn} & {3.68\\(2.48 $\times$)} & {13.37\\(1.61 $\times$)} & {1.68\\(1.47 $\times$)}  
\end{tblr}
\end{table}

To evaluate the efficiency improvement of non-linear math functions brought from the new MSNZB protocol, we also test it on $1/x$, $1/\sqrt{x}$ and $\ln x$. 
We test $10^5$ numbers under both LAN and WAN settings.
We set the $1/x$ and $1/\sqrt{x}$ in~\cite{rathee2021sirnn}, and $\ln x$ in~\cite{lu2020faster} as baselines. 
As summarized in \autoref{tab: funcs}, our method reduces communication cost by about $1.15\times$, and about $1.70\times$ faster execution under LAN settings, and about $1.25\times$ under WAN conditions.

\begin{table}[!ht]
\centering
\caption{Performance comparison against ~\cite{rathee2021sirnn} and ~\cite{lu2020faster} under different network conditions.}
\label{tab: funcs}
\begin{tblr}{
  cells = {c},
  colsep = 3pt,
  cell{1}{1} = {c=2,r=2}{},
  cell{1}{3} = {c=2}{},
  cell{1}{5} = {r=2}{},
  cell{1}{6} = {r=2}{},
  cell{3}{1} = {r=2}{},
  cell{5}{1} = {r=2}{},
  cell{7}{1} = {r=2}{},
  vline{3-6} = {3-8}{},
  hline{1,3,9} = {-}{0.08em},
  hline{4,6,8} = {2-6}{},
  hline{5,7} = {-}{},
}
Method       &                                         & Total Time (in sec)     &                          & {Comm./Inst. \\ (in KB)} & {Max ULP\\ Error} \\
             &                                         & LAN                     & WAN                      &                          &                   \\
1/x          & \sysname                 & 2.86                    & 18.97                    & 3.37                     & 2                 \\
             & \cite{rathee2021sirnn} & {5.06\\(1.77 $\times$)} & {24.01\\(1.26 $\times$)} & {3.91\\(1.16 $\times$)}  & 1                 \\
$1/\sqrt{x}$ & \sysname                 & 3.22                    & 21.78                    & 3.56                     & 4                 \\
             & \cite{rathee2021sirnn} & {5.42\\(1.68 $\times$)} & {26.82\\(1.23 $\times$)} & {4.10\\(1.15 $\times$)}  & 4                 \\
$\ln x$      & \sysname                 & 2.96                    & 20.44                    & 3.29                     & 12                \\
             & \cite{lu2020faster}    & {5.16\\(1.74 $\times$)} & {25.48\\(1.25 $\times$)} & {3.84\\(1.16 $\times$)}  & 28                
\end{tblr}
\end{table}

\subsection{Case studies}
\label{text: case}
\begin{table*}[htbp]
\centering
\caption{Case studies of \sysname and baseline implemented on EzPC}
\label{tab: case_study}
\centering
\begin{tblr}{
  cells = {c},
  cell{1}{1} = {r=2}{},
  cell{1}{2} = {c=2}{},
  cell{1}{4} = {c=2}{},
  cell{1}{6} = {c=2}{},
  vline{3,5} = {1}{0.08em},
  vline{4,6} = {2}{0.08em},
  vline{2-7} = {3-15}{},
  hline{1,3,16} = {-}{0.08em},
  hline{2} = {2-7}{0.08em},
  hline{4-15} = {-}{},
}
Functions         & Runtime(ms) LAN &       & Runtime(ms) WAN &         & Communication(KB) &         \\
                  & Baseline            & \sysname  & Baseline       & \sysname    & Baseline              & \sysname    \\
Mode              & 263.12          & 195.66    & 23633.80        & 20897.83        & 14864.63          & 4332.40        \\
Median/Percentile & 764.43          & 41.83     & 9915.76         & 1450.39         & 3149.44           & 2688.69        \\
Rank              & 3064.60         & 77.58     & 321953.92      & 1567.40         & 16709.84          & 4865.48        \\
Min/Max           & 796.87          & 36.17     & 123398.53       & 1048.07         & 9879.38           & 2645.67        \\
Std               & 0.07            & 0.05      & 5.75            & 5.70            & 13.45             & 12.91          \\
Skew              & 0.10            & 0.08      & 14.89           & 14.84           & 32.13             & 31.59          \\
Kurt              & 0.09            & 0.07      & 11.23           & 11.18           & 22.78             & 22.24          \\
Linest            & 0.10            & 0.08      & 14.86           & 14.81           & 31.94             & 31.40          \\
Logest            & 61.78           & 35.46     & 319.36          & 259.06          & 4618.36           & 3971.73        \\
Pearson           & 0.12            & 0.10      & 20.36           & 20.31           & 41.46             & 40.92          \\
T.Test            & 0.12            & 0.08      & 5.99            & 5.89            & 17.36             & 16.28          \\
CHISQ.Test        & 54.95           & 36.67     & 832.67          & 423.15          & 6741.99           & 1621.38        \\
F.Test            & 0.20            & 0.14      & 11.78           & 11.63           & 31.29             & 29.67       
\end{tblr}
\end{table*}

To evaluate the efficiency of \sysname for statistical analysis, we select $14$ commonly used statistical functions.
All selected functions can be constructed using the three categories of secure primitives provided by \sysname.
For example, consider the secure $\mathsf{CHISQ.Test}$, defined as $\chi^2=\sum_{i=1}^{a}\sum_{j=1}^{b}\frac{(O_{ij}-E_{ij})^2}{E_{ij}}$, where $O_{ij}$ denotes the observed frequency and $E_{ij} = \frac{O_i O_j}{N}$.
This function can be implemented by combining the secure frequency counting protocol $\Pi_{\text{PFC}}^{m,n}$, the secure reciprocal ($1/x$) primitive, and secure multiplication.
In this subsection, we report experimental results obtained using EzPC.
For baseline comparison, we also implement the same statistical functions using the corresponding secure primitives.
Specifically, for secure frequency counting, we adopt a straightforward equality-test-based implementation as the baseline.
For secure sorting, we employ the \emph{shuffle-then-compare} paradigm from~\cite{hamada2012practically}, instantiated with the SOTA secure shuffle protocol~\cite{chase2020secret} and secure comparison~\cite{rathee2020cryptflow2, huang2022cheetah}.
For secure non-linear mathematical functions, we use the protocols from~\cite{rathee2021sirnn} as the baseline.

We conduct performance evaluations using the Students Academic Performance Evaluation Dataset\footnote{https://data.mendeley.com/datasets/dc3797vf3t/1}, a publicly available educational dataset comprising approximately $1195$ students with multiple learning behavior features and corresponding academic performance indicators.
In our experiments, we use average class attendance and cumulative grade point average (CGPA) as input attributes, with domain sizes of $100$ and $400$, respectively.
For single-variable functions, average class attendance is used as the input.
For two-variable functions, both average class attendance and CGPA are used.
For the $\mathsf{CHISQ.Test}$, we discretize average class attendance and CGPA into $5$ and $10$ categories, respectively.

As shown in \autoref{tab: case_study}, functions implemented using \sysname consistently outperform the baseline across most evaluated cases.
The performance gains are particularly pronounced for functions that rely on secure sorting, such as $\mathsf{Rank}$, which achieves a $39.5\times$ speedup in running time under the LAN setting and a $205.4\times$ speedup under the WAN setting, along with a $3.4\times$ reduction in communication.
In contrast, the improvement is less pronounced for multiplication-intensive statistical functions such as $\mathsf{Std}$, where the dominant cost arises from secure multiplication and both \sysname and the baseline use the same implementation provided by EzPC.
Overall, these results suggest that \sysname is particularly effective in accelerating secure statistical computations that involve non-linear operations, while maintaining competitive performance across a broad range of data statistics-related functions.

\section{Conclusion}

In this paper, we propose \sysname, a secure and efficient toolkit designed for 2PC-based statistical analysis. 
By examining Microsoft Excel’s statistical library, we observed that a wide range of statistical tasks can be decomposed into three fundamental operations: frequency counting, sorting, and non-linear mathematical functions. 
Guided by this insight, \sysname provides specialized 2PC implementations for each core operation, enabling secure statistical analysis with significantly improved efficiency.
For frequency counting, \sysname introduces a secure shift-based strategy that avoids costly equality tests. 
For sorting, we propose a secure segment-indicator protocol that supports counting-based sorting without relying on secure comparisons, making it particularly suitable for common statistical scenarios. 
For non-linear math functions, \sysname advances the reduce-then-approximate paradigm with a bisection-based range reduction protocol, improving both accuracy and efficiency.
Our experimental evaluation on 14 representative statistical analysis tasks demonstrates the efficiency improvement of \sysname to implement data statistic-related functions.

\bibliographystyle{IEEEtran}
\bibliography{ref}

\appendix

\subsection{Security Proof}
Due to the page limit, we include the complete simulation proof only for $\Pi_\text{PFC}^{m,n}$, our secure frequency counting protocol $\Pi_\text{PFC}^{m,n}$ and secure sorting protocol $\Pi_\text{CSort}^{m,n}$.
The other protocols are proved using the same simulation paradigm together with UC composition in their respective hybrid models, so we only provide concise proof sketches here. The full simulators and detailed indistinguishability arguments are deferred to the full version.

\subsubsection{Secure Frequency Counting}
\label{sec: pfc}

\begin{proof}
    To prove Theorem~\ref{theorem: pfc}, we construct a PPT simulator $\mathcal{S}$, such that no non-uniform PPT environment $\mathcal{Z}$ can distinguish between the ideal world $\mathsf{Ideal}_{\mathcal{F}_\text{PFC}^{m,n},\mathcal{S},\mathcal{Z}}(1^\lambda) $ and the real world $\mathsf{Real}_{ \Pi_\text{PFC}^{m,n},\mathcal{A},\mathcal{Z}}^{\mathcal{F}_\text{RVOSE}, \mathcal{F}_\text{B2A}}(1^\lambda)$. We consider the following cases:

    $\mathsf{Case}$ $1$: $P_0$ is corrupted.
    We construct the simulator $\mathcal{S}$ that internally runs $\mathcal{A}$, simulates $\mathcal{F}_\text{RVOSE}$ and $\mathcal{F}_\text{B2A}$, forwards messages to/from $\mathcal{Z}$, and simulates the interface of the honest party $P_1$.

    Upon receiving $(\mathsf{Input}, \mathsf{sid}, P_1)$ from $\mathcal{F}_\text{PFC}^{m,n}$, $\mathcal{S}$ does as follows.
    \begin{itemize}
        \item Upon receiving $\mathsf{Output}, \mathsf{sid}, \langle \vec{y}\rangle_0$ from $\mathcal{F}_\text{PFC}^{m,n}$, $\mathcal{S}$ records $\langle \vec{y}\rangle_0$.
        \item  When $P_0$ invokes $\mathcal{F}_\text{RVOSE}$ with $(\mathsf{input, \mathsf{sid}, \vec{a_i}})$, $\mathcal{S}$  records $\vec{a_i}$ and samples $\vec{\check{b_i}} \gets \mathbb{Z}_{2}^m$, $\check{\epsilon_i} \gets \mathbb{Z}_{m}$, computes $\vec{\check{c_i}}:=\text{shift}(\check{\epsilon_i},\vec{a_i}) \oplus \vec{\check{b_i}}$ and sends $\vec{\check{b_i}}$ to $P_0$ act as $\mathcal{F}_\text{RVOSE}$. 
        \item $\mathcal{S}$ samples $r_i \gets \mathbb{Z}_{n}$ and sends $r_i$ to $P_0$ act as $P_1$.
        \item When $P_0$ invokes $\mathcal{F}_\text{B2A}$ with $(\mathsf{input}, \mathsf{sid}, \langle\vec{y''_i}\rangle_0)$, $\mathcal{S}$ samples $\langle \vec{\check{y'_i}}\rangle_0 \gets \mathbb{Z}_{2^{\ell_n}}^m$, for $i \in [n-1]$, computes $\langle \vec{\check{y'_{n-1}}}\rangle_0 := \langle \vec{y}\rangle_0 - \sum_{i=0}^{n-2}\langle \vec{\check{y'}}\rangle_0$, and sends $\langle \vec{\check{y'_i}}\rangle_0$ to $P_0$ act as $\mathcal{F}_\text{B2A}$. 
    \end{itemize}

    \noindent $\mathsf{Indistinguishability.}$ We show that the incoming message and the output of $P_0$ in the ideal world are indistinguishable from the real world. 

    \begin{claim}
     The ideal world $\mathsf{Ideal}_{\mathcal{F}_\text{PFC}^{m,n},\mathcal{S},\mathcal{Z}}(1^\lambda) $ and the real world $\mathsf{Real}_{ \Pi_\text{PFC}^{m,n},\mathcal{A},\mathcal{Z}}^{\mathcal{F}_\text{RVOSE},\mathcal{F}_\text{B2A}}(1^\lambda)$ are perfectly indistinguishable.
    \end{claim}
  
    \begin{proof}
    There are three parts of incoming messages that are different between $\mathsf{Ideal}_{\mathcal{F}_\text{PFC}^{m,n},\mathcal{S},\mathcal{Z}}(1^\lambda)$ and $\mathsf{Real}_{ \Pi_\text{PFC}^{m,n},\mathcal{A},\mathcal{Z}}^{\mathcal{F}_\text{RVOSE},\mathcal{F}_\text{B2A}}(1^\lambda)$.

    \begin{itemize}  
        \item The output of $\mathcal{F}_\text{RVOSE}$ is a random boolean vector, which is indistinguishable between the ideal world and the real world.
        \item In the ideal world, $r_i$ are sampled uniformly at random rather than real $\langle \vec{x}[i]\rangle_1 - \epsilon_i$.
        \item In the ideal world, $\langle \vec{y'}[i]\rangle_0, i \in [n-1]$ are sampled uniformly at random, and $\langle \vec{\check{y'}}[n-1]\rangle_0$ are calculated with $\langle \vec{y}\rangle_0 - \sum_{i=0}^{n-2}\langle \vec{\check{y_i'}}\rangle$ rather than the output of $\mathcal{F}_\text{B2A}$. 
    \end{itemize}

    For the second part, due to the $\epsilon_i$ in the real world are sampled uniformly at random, $\langle \vec{x}[i]\rangle_1 - \epsilon \bmod n$ are also uniformly random, which is indistinguishable from the ideal world.
    For the third part, $\langle \vec{y'_i} \rangle_0$ are the output of $\mathcal{F}_\text{B2A}$, which are random arithmetic vectors. $\langle \vec{\check{y'_i}}\rangle_0, i\in [n-2]$ are sampled uniformly at random, and $\langle \vec{\check{y'_{n-1}}}\rangle_0 := \langle \vec{y}\rangle_0 - \sum_{i=0}^{n-2}\langle \vec{\check{y_i'}}\rangle$ is also uniformly random, thus $\langle \vec{\check{y'_i}}\rangle_0$ in the ideal world are indistinguishable with $\langle \vec{y'_i}\rangle_0$ in the real world.
    \end{proof}
   
    $\mathsf{Case}$ $2$: $P_1$ is corrupted.
    We construct the simulator $\mathcal{S}$ that internally runs $\mathcal{A}$, simulates $\mathcal{F}_\text{RVOSE}$ and $\mathcal{F}_\text{B2A}$, forwards messages to/from $\mathcal{Z}$, and simulates the interface of the honest party $P_0$.

    Upon receiving $(\mathsf{Input}, \mathsf{sid}, P_0)$ from $\mathcal{F}_\text{PFC}^{m,n}$, $\mathcal{S}$ does as follows.
    \begin{itemize}
        \item Upon receiving $\mathsf{Output}, \mathsf{sid}, \langle \vec{y}\rangle_0$ from $\mathcal{F}_\text{PFC}^{m,n}$, $\mathcal{S}$ records $\langle \vec{y}\rangle_1$.
        \item When $P_1$ invokes $\mathcal{F}_\text{RVOSE}$, $\mathcal{S}$ samples $\vec{\check{c_i}} \gets \mathbb{Z}_{2}^m$, $\check{\epsilon_i} \gets \mathbb{Z}_{m}$, and sends $\vec{\check{c_i}}$ and $\check{\epsilon_i}$ to $P_1$ act as $\mathcal{F}_\text{RVOSE}$. 
        \item $\mathcal{S}$ samples $\vec{\check{d_i}} \gets \mathbb{Z}_{2}^n$ and sends $\vec{\check{d_i}}$ to $P_1$ act as $P_0$.
        \item When $P_1$ invokes $\mathcal{F}_\text{B2A}$ with $(\mathsf{input}, \mathsf{sid}, \langle\vec{y''_i}\rangle_1)$, $\mathcal{S}$ samples $\langle \vec{\check{y'_i}}\rangle_1 \gets \mathbb{Z}_{2^{\ell_n}}^m$, for $i \in [n-1]$, computes $\langle \vec{\check{y'_{n-1}}}\rangle_1 := \langle \vec{y}\rangle_1 - \sum_{i=0}^{n-2}\langle \vec{\check{y'}}\rangle_1$, and sends $\langle \vec{\check{y'_i}}\rangle_1$ to $P_1$ act as $\mathcal{F}_\text{B2A}$. 
    \end{itemize}

    \noindent $\mathsf{Indistinguishability.}$ We show that the incoming message and the output of $P_1$ in the ideal world are indistinguishable from the real world. 

    \begin{claim}
     The ideal world $\mathsf{Ideal}_{\mathcal{F}_\text{PFC}^{m,n},\mathcal{S},\mathcal{Z}}(1^\lambda) $ and the real world $\mathsf{Real}_{ \Pi_\text{PFC}^{m,n},\mathcal{A},\mathcal{Z}}^{\mathcal{F}_\text{RVOSE},\mathcal{F}_\text{B2A}}(1^\lambda)$ are perfectly indistinguishable.
    \end{claim}
  
    \begin{proof}
    There are three parts of incoming messages that are different between $\mathsf{Ideal}_{\mathcal{F}_\text{PFC}^{m,n},\mathcal{S},\mathcal{Z}}(1^\lambda)$ and $\mathsf{Real}_{ \Pi_\text{PFC}^{m,n},\mathcal{A},\mathcal{Z}}^{\mathcal{F}_\text{RVOSE},\mathcal{F}_\text{B2A}}(1^\lambda)$.

    \begin{itemize}  
        \item The output of $\mathcal{F}_\text{RVOSE}$ is a random boolean vector, which is indistinguishable between the ideal world and the real world.
        \item In the ideal world, $\langle \vec{\check{d_i}}\rangle_1$ are sampled uniformly at random rather than real $\langle \vec{d_i}\rangle_1 := \vec{e_i}\oplus \vec{a_i}$.
        \item In the ideal world, $\langle \vec{y'}[i]\rangle_1, i \in [n-1]$ are sampled uniformly at random, and $\langle \vec{\check{y'}}[n-1]\rangle_1$ are calculated with $\langle \vec{y}\rangle_1 - \sum_{i=0}^{n-2}\langle \vec{\check{y_i'}}\rangle_1$ rather than the output of $\mathcal{F}_\text{B2A}$. 
    \end{itemize}

    For the second part, due to the $\vec{a_i}$ in the real world are the output of $\mathcal{F}_\text{RVOSE}$ which are uniformly at random, such that it is indistinguishable from $\vec{\check{d_i}}$ in the ideal world.
    For the third part, the indistinguishability is same as the $\mathsf{Case}$ $1$.
    \end{proof}

This concludes the proof.
\end{proof}

\subsubsection{Secure Counting Sort}
\label{sec: csort}

\begin{proof}
    To prove Theorem~\ref{theorem: csort}, we construct a PPT simulator $\mathcal{S}$, such that no non-uniform PPT environment $\mathcal{Z}$ can distinguish between the ideal world $\mathsf{Ideal}_{\mathcal{F}_\text{sort}^{m,n},\mathcal{S},\mathcal{Z}}(1^\lambda) $ and the real world $\mathsf{Real}_{ \Pi_\text{CSort}^{m,n},\mathcal{A},\mathcal{Z}}^{\mathcal{F}_\text{PFC}^{m,n}, \mathcal{F}_\text{convert}^{N_1 \rightarrow N_2}, \mathcal{F}_\text{RVOSE}}(1^\lambda)$. 
    We consider the following cases:

    $\mathsf{Case}$ $1$: $P_0$ is corrupted.
    We construct the simulator $\mathcal{S}$ that internally runs $\mathcal{A}$, simulates $\mathcal{F}_\text{convert}^{N_1\rightarrow N_2}$, $\mathcal{F}_\text{PFC}^{m,n}$, and $\mathcal{F}_\text{RVOSE}$, forwards messages to/from $\mathcal{Z}$, and simulates the interface of the honest party $P_1$.

    Upon receiving $(\mathsf{Input}, \mathsf{sid}, P_1)$ from $\mathcal{F}_\text{sort}^{m,n}$, $\mathcal{S}$ does as follows. 
    \begin{itemize}
        \item Upon receiving $\mathsf{Output}, \mathsf{sid}, \langle \vec{y}\rangle_0$ from $\mathcal{F}_\text{sort}^{m,n}$, $\mathcal{S}$ records $\langle \vec{y}\rangle_0$, and samples $\langle \vec{\check{y''_i}}\rangle_0 \gets \mathbb{Z}_{2^{\ell_m}}^n, i\in[1,m-1]$, computes $ \langle \vec{\check{y''_0}}\rangle_0 := 2\cdot\langle \vec{\check{y''_1}}\rangle_0 + \sum_{i=2}^{m-1}(\langle \vec{\check{y''_i}}\rangle_0 - \langle \vec{\check{y''_{i-1}}}\rangle_0)\cdot(i+1) - \langle \vec{y}\rangle_0  - \vec{1}$, samples $r_i \gets \mathbb{Z}_{n+1}$.
        $\mathcal{S}$ computes $\langle \vec{\check{y_i'''}}\rangle_0$ by solving the equation $\langle \vec{\check{y''_i}}\rangle_0 = A \langle \vec{\check{y'''_i}}\rangle_0$, where $A \in \mathbb{Z}_{2^{\ell_m}}^{n\times n}$ is an  unit lower triangular matrix.
        \item When $P_0$ invokes $\mathcal{F}_\text{RVOSE}$ with $\vec{a_i}$, $\mathcal{S}$ sets $\vec{\check{b_i}} := \text{shift}(-r_i, \langle \vec{\check{y'''_{i}}}\rangle_0)$, and sends $\vec{\check{b_i}}$ to $P_0$ act as $\mathcal{F}_\text{RVOSE}$.
        \item When $P_0$ invokes $\mathcal{F}_\text{PFC}^{m,n}$ with $\langle \vec{x}\rangle_0$, $\mathcal{S}$ samples $\langle \vec{\check{q''}}\rangle \gets \mathbb{Z}_{2^{\ell_n}}^{m}$, and sends $\langle \vec{\check{q''}}\rangle$ to $P_0$ act as $\mathcal{F}_\text{PFC}^{m,n}$.
        \item When $P_0$ invokes $\mathcal{F}_\text{convert}^{2^{\ell_n}\rightarrow n+1}$ with $\langle \vec{q''}\rangle^{2^{\ell_n}}_0$, $\mathcal{S}$ samples $\langle \vec{\check{q'}}\rangle_0 \gets \mathbb{Z}_{2^{n+1}}^{m}$, 
        \item To simulate the $\langle \vec{q}[i]\rangle_1-\epsilon_i$, $\mathcal{S}$ sends $r_i$ to $P_0$ act as $P_1$.
    \end{itemize}

    \noindent $\mathsf{Indistinguishability.}$ We show that the incoming message and the output of $P_0$ in the ideal world are indistinguishable from the real world. 

    \begin{claim}
     The ideal world $\mathsf{Ideal}_{\mathcal{F}_\text{sort}^{m,n},\mathcal{S},\mathcal{Z}}(1^\lambda) $ and the real world $\mathsf{Real}_{ \Pi_\text{CSort}^{m,n},\mathcal{A},\mathcal{Z}}^{\mathcal{F}_\text{RVOSE},\mathcal{F}_\text{convert}^{N_1\rightarrow N_2}, \mathcal{F}_\text{PFC}^{m,n}}(1^\lambda)$ are perfectly indistinguishable.
    \end{claim}
  
    \begin{proof}
    There are four parts of incoming messages that are different between $\mathsf{Ideal}_{\mathcal{F}_\text{sort},\mathcal{S},\mathcal{Z}}(1^\lambda)$ and $\mathsf{Real}_{ \Pi_\text{sort},\mathcal{A},\mathcal{Z}}^{\mathcal{F}_\text{RVOSE},\mathcal{F}_\text{PFC}^{m,n}, \mathcal{F}_\text{convert}^{N_1 \rightarrow N_2}}(1^\lambda)$.

    \begin{itemize} 
        \item In the ideal world, $\langle \vec{\check{b_i}}\rangle_0$ are calculated based on $\langle \vec{\check{y'''_i}}\rangle_0$ rather than the output of $\mathcal{F}_\text{RVOSE}$.
        \item The output of $\mathcal{F}_\text{PFC}^{m,n}$ is a random arithmetic vector, which is indistinguishable between the ideal world and the real world.
        \item The output of $\mathcal{F}_\text{convert}^{2^{\ell_n}\rightarrow n+1}$ is a random arithmetic vector, which is indistinguishable between the ideal world and the real world.
        \item In the ideal world, $r_i$ are sampled randomly rather than $\langle \vec{q}[i]\rangle_1 - \epsilon_i$.
    \end{itemize}

    For the first part, based on the randomness of $\langle \vec{\check{y''_i}}\rangle_0, \in [1, m-1]$, $\langle \vec{\check{y_0''}}\rangle_0$ is a random vector. Given the matrix $A$ is an unit lower triangular matrix, the map from $\vec{s}$ to $\vec{t}=A\vec{s}$ is bijection. Such that $\langle \vec{\check{y_i'''}}\rangle$ is a random vector.
    $\vec{\check{b_i}}$ is circular shift of $\langle \vec{\check{y_i'''}}\rangle_0$, which is indistinguishable from the output of $\mathcal{F}_\text{RVOSE}$ in the real world.
    For the fourth part, due to the $\epsilon_i$ in the real world are uniformly random, $\langle \vec{q}[i]\rangle_1 - \epsilon \bmod (n+1)$ are uniformly random, which is indistinguishable from the ideal world.
    \end{proof}

    $\mathsf{Case}$ $2$: $P_1$ is corrupted.
    We construct the simulator $\mathcal{S}$ that internally runs $\mathcal{A}$, simulates $\mathcal{F}_\text{convert}^{N_1\rightarrow N_2}$, $\mathcal{F}_\text{PFC}^{m,n}$, and $\mathcal{F}_\text{RVOSE}$, forwards messages to/from $\mathcal{Z}$, and simulates the interface of the honest party $P_0$.

    Upon receiving $(\mathsf{Input}, \mathsf{sid}, P_0)$ from $\mathcal{F}_\text{sort}^{m,n}$, $\mathcal{S}$ does as follows.
    \begin{itemize}
        \item Upon receiving $(\mathsf{Output}, \mathsf{sid}, \langle \vec{y}\rangle_1)$ from $\mathcal{F}_\text{sort}^{m,n}$, $\mathcal{S}$ records $\langle \vec{y}\rangle_1$, and samples $\langle \vec{\check{y''_i}}\rangle_1 \gets \mathbb{Z}_{2^{\ell_m}}^n, i\in[1,m-1]$, computes $ \langle \vec{\check{y''_0}}\rangle_1 := 2\cdot\langle \vec{\check{y''_1}}\rangle_1 + \sum_{i=2}^{m-1}(\langle \vec{\check{y''_i}}\rangle_1 - \langle \vec{\check{y''_{i-1}}}\rangle_1)\cdot(i+1) - \langle \vec{y}\rangle_1  - \vec{1}$, samples $r_i \gets \mathbb{Z}_{n+1}$.
        $\mathcal{S}$ computes $\langle \vec{\check{y_i'''}}\rangle_1$ by solving the equation $\langle \vec{\check{y''_i}}\rangle_1 = A \langle \vec{\check{y'''_i}}\rangle_1$, where $A \in \mathbb{Z}_{2^{\ell_m}}^{n\times n}$ is an  unit lower triangular matrix.
        \item When $P_1$ invokes $\mathcal{F}_\text{RVOSE}$, $\mathcal{S}$ sets $\vec{\check{c_i}} \gets \mathbb{Z}_{2^{\ell_m}}^{n+1}, \check{\epsilon}_i \gets \mathbb{Z}_{n+1}$ , and sends $\vec{\check{c_i}}, \check{\epsilon_i}$ to $P_1$ act as $\mathcal{F}_\text{RVOSE}$.
        \item When $P_1$ invokes $\mathcal{F}_\text{PFC}^{m,n}$ with $\langle \vec{x}\rangle_1$, $\mathcal{S}$ samples $\langle \vec{\check{q''}}\rangle_1 \gets \mathbb{Z}_{2^{\ell_n}}^{m}$, and sends $\langle \vec{\check{q''}}\rangle_1$ to $P_1$ act as $\mathcal{F}_\text{PFC}^{m,n}$.
        \item When $P_1$ invokes $\mathcal{F}_\text{convert}^{2^{\ell_n}\rightarrow n+1}$ with $\langle \vec{q''}\rangle^{2^{\ell_n}}_1$, $\mathcal{S}$ samples $\langle \vec{\check{q'}}\rangle_1 \gets \mathbb{Z}_{n+1}^{m}$,  sends $\langle \vec{\check{q'}}\rangle_1$ to $P_1$ act as $\mathcal{F}_\text{convert}^{2^{\ell_n}\rightarrow n+1}$, and computes $\langle \vec{\check{q}}[i]\rangle := \sum_{j=0}^i(\langle \vec{\check{q'}}[j] \rangle)$.
        \item To simulate the $\vec{d}_i$, $\mathcal{S}$ computes $\vec{\check{t_i}}:= \text{shift}(\check{\epsilon_i} - \langle \vec{\check{q}}[i]\rangle_1,\langle \vec{\check{y_i'''}}\rangle_1)$, $\vec{\check{d_i}}:=\text{shift}(-\check{\epsilon_i}, \vec{\check{t_i}} - \vec{\check{c_i}})$ sends $\vec{\check{d_i}}$ to $P_1$ act as $P_0$.
    \end{itemize}

    \noindent $\mathsf{Indistinguishability.}$ We show that the incoming message and the output of $P_0$ in the ideal world are indistinguishable from the real world. 

    \begin{claim}
     The ideal world $\mathsf{Ideal}_{\mathcal{F}_\text{sort}^{m,n},\mathcal{S},\mathcal{Z}}(1^\lambda) $ and the real world $\mathsf{Real}_{ \Pi_\text{CSort}^{m,n},\mathcal{A},\mathcal{Z}}^{\mathcal{F}_\text{RVOSE},\mathcal{F}_\text{convert}^{2^{\ell_n}\rightarrow n+1}, \mathcal{F}_\text{PFC}^{m,n}}(1^\lambda)$ are perfectly indistinguishable.
    \end{claim}
  
    \begin{proof}
    There are four parts of incoming messages that are different between $\mathsf{Ideal}_{\mathcal{F}_\text{PFC}^{m,n},\mathcal{S},\mathcal{Z}}(1^\lambda)$ and $\mathsf{Real}_{ \Pi_\text{PFC}^{m,n},\mathcal{A},\mathcal{Z}}^{\mathcal{F}_\text{RVOSE},\mathcal{F}_\text{B2A}}(1^\lambda)$.

    \begin{itemize}  
        \item The output of $\mathcal{F}_\text{RVOSE}$ is a random arithmetic vector and a random number, which are indistinguishable between the ideal world and the real world.
        \item The output of $\mathcal{F}_\text{PFC}^{m,n}$ is a random arithmetic vector, which is indistinguishable between the ideal world and the real world.
        \item The output of $\mathcal{F}_\text{convert}^{2^{\ell_n}\rightarrow n+1}$ is a random arithmetic vector, which is indistinguishable between the ideal world and the real world.
        \item In the ideal world, $\vec{\check{d_i}}$ are calculated based on $\langle \vec{\check{y'''_i}}\rangle$ rather than $\vec{d_i}$ in the real world.
    \end{itemize}

    For the fourth part, we can proof $\langle \vec{\check{y'''_i}}\rangle$ are uniformly random in a similar way in $\mathsf{Case}$ $1$, then $\vec{\check{t_i}}$ is uniformly random, which is a circular vector shift of  $\langle \vec{\check{y'''_i}}\rangle$. 
    $\vec{\check{d_i}} =\text{shift}(-\check{\epsilon_i}, \vec{\check{t_i}} - \vec{\check{c_i}})$ is uniformly random based on the randomness of $\vec{\check{t_i}}$.
    Then, the fourth part is indistinguishable between the ideal world and the real world.
    \end{proof}

This concludes the proof.
\end{proof}

\subsubsection{Secure MSNZB}
\label{sec: msnzb}
\begin{proof}
    To prove Theorem~\ref{theorem: msnzb}, we construct a PPT simulator $\mathcal{S}$ to simulate the views of the corrupted party.

    $\mathsf{Case}$ $1$: $P_0$ is corrupted.
    In the hybrid model, $P_0$ receives messages from the functionalities, $\mathcal{F}_\text{Mill}$,$\mathcal{F}_\text{B2A}$,$\mathcal{F}_\text{eq}$,$\mathcal{F}_\text{MUX}$,$\mathcal{F}_\text{LUT}$,$\mathcal{F}_\text{B2A}$ and $\mathcal{F}_\text{one-hot}$.
    By the UC security of the functionalities, their outputs are uniformly random.
    We construct the simulator $\mathcal{S}$ to samples uniformly random values to simulate the messages that $P_0$ receives.

    $\mathsf{Case}$ $2$: $P_1$ is corrupted.
    The construction of $\mathcal{S}$ and the indistinguishability are the same as $\mathsf{Case}$ $1$.
\end{proof}

\subsubsection{Secure Mode}
\label{sec: mode}
\myherebox{Functionality $\mathcal{F}_\text{mode}^{m,n,N}$}{white!20}{white!10}{
  $\mathcal{F}_\text{mode}$ interacts with $P_0$, $P_1$ and the simulator $\mathcal{S}$.

\underline{\textbf{Parameters:}} $n$ is the size of the dataset; $m$ is the size of feature domain, $N$ is an integer sightly larger than $m$.
  
\underline{\textbf{Input:}}
  \begin{itemize}[leftmargin=*]
  \item Upon receiving $(\mathsf{Input}, \mathsf{sid}, \langle \vec{x}\rangle_0)$ from $P_0$, record $\langle \vec{x}\rangle_0$ and send $(\mathsf{Input}, \mathsf{sid}, P_0)$ to $\mathcal{S}$, where $\langle \vec{x}\rangle_0 \in \mathbb{Z}_m^n$.
  \item Upon receiving $(\mathsf{Input}, \mathsf{sid}, \langle \vec{x}\rangle_1)$ from $P_1$, record $\langle \vec{x}\rangle_1$ and send $(\mathsf{Input}, \mathsf{sid}, P_1)$ to $\mathcal{S}$, where $\langle \vec{x}\rangle_1 \in \mathbb{Z}_m^n$.
  \end{itemize}
    
\underline{\textbf{Execution:}}
  \begin{itemize}[leftmargin=*]
   \item If both $\langle \vec{x}\rangle_0, \langle \vec{x}\rangle_1$ are recorded, reveal $\vec{x} := \langle \vec{x}\rangle_0 + \langle \vec{x}\rangle_1$, computes the mode value $y$, samples $y_0 \gets \mathbb{Z}_{N}$, and computes $y_1 := y - y_0$.
   \item Send $(\mathsf{Output},\mathsf{sid}, y_b)$ to $P_b$, $b \in \{0, 1\}$.
  \end{itemize}
}{The Ideal Functionality $\mathcal{F}_\text{mode}$.}{func: mode}

We define the functionality $\mathcal{F}_\text{percentile}^{m,n,k}$ for secure percentile in \autoref{func: percentile} and describe the protocol in \autoref{algo: exkth}.

\begin{theorem}
\label{theorem: mode}
$\Pi_\text{percentile}^{m,n,k}$ UC realizes $\mathcal{F}_\text{percentile}^{m,n,k}$ in the $(\mathcal{F}_\text{PFC}^{m,n}, \mathcal{F}_\text{RVOSE})$ hybrid model against semi-honest probabilistic polynomial time (PPT) adversaries with statistical corruption.
\end{theorem}

\begin{proof}
    To prove Theorem~\ref{theorem: mode}, we construct a PPT simulator $\mathcal{S}$ to simulate the views of the corrupted party.

    $\mathsf{Case}$ $1$: $P_0$ is corrupted.
    In the hybrid model, $P_0$ receives messages from the functionalities, $\mathcal{F}_\text{PFC}^{m,n}, \mathcal{F}_\text{convert}^{N_1\rightarrow N_2}, \mathcal{F}_\text{argmax}^{t_i}$.
    By the UC security of the functionalities, their outputs are uniformly random. 
    We construct the simulator $\mathcal{S}$ to samples uniformly random values to simulate the messages that $P_0$ receives.

    $\mathsf{Case}$ $2$: $P_1$ is corrupted.
    The construction of $\mathcal{S}$ and the indistinguishability are the same as $\mathsf{Case}$ $1$.
\end{proof}

\subsubsection{Secure Percentile}
\label{sec: kth}
\myherebox{Functionality $\mathcal{F}_\text{percentile}^{m,n,k}$}{white!20}{white!10}{
  $\mathcal{F}_\text{percentile}$ interacts with $P_0$, $P_1$ and the simulator $\mathcal{S}$.

\underline{\textbf{Parameters:}} $n$ is the size of the dataset; $m$ is the size of feature domain, $k$ is the corresponding index of the $\mathsf{percentile}$ value in the sorted vector.
  
\underline{\textbf{Input:}}
  \begin{itemize}[leftmargin=*]
  \item Upon receiving $(\mathsf{Input}, \mathsf{sid}, \langle \vec{x}\rangle_0)$ from $P_0$, record $\langle \vec{x}\rangle_0$ and send $(\mathsf{Input}, \mathsf{sid}, P_0)$ to $\mathcal{S}$, where $\langle \vec{x}\rangle_0 \in \mathbb{Z}_m^n$.
  \item Upon receiving $(\mathsf{Input}, \mathsf{sid}, \langle \vec{x}\rangle_1)$ from $P_1$, record $\langle \vec{x}\rangle_1$ and send $(\mathsf{Input}, \mathsf{sid}, P_1)$ to $\mathcal{S}$, where $\langle \vec{x}\rangle_1 \in \mathbb{Z}_m^n$.
  \end{itemize}
    
\underline{\textbf{Execution:}}
  \begin{itemize}[leftmargin=*]
   \item If both $\langle \vec{x}\rangle_0, \langle \vec{x}\rangle_1$ are recorded, reveal $\vec{x} := \langle \vec{x}\rangle_0 + \langle \vec{x}\rangle_1$, computes the sorted vector $\vec{y}$, uniformly random sample $y_0 \gets \mathbb{Z}_{2^{\ell_m}}$, and compute $y_1 := \vec{y}[k] - y_0$.
   \item Send $(\mathsf{Output},\mathsf{sid}, y_b)$ to $P_b$, $b \in \{0, 1\}$.
  \end{itemize}
}{The Ideal Functionality $\mathcal{F}_\text{percentile}^{m,n,k}$.}{func: percentile}

We define the functionality $\mathcal{F}_\text{percentile}^{m,n,k}$ for secure percentile in \autoref{func: percentile} and describe the protocol in \autoref{algo: exkth}. 

\begin{theorem}
\label{theorem: kth}
$\Pi_\text{percentile}^{m,n,k}$ UC realizes $\mathcal{F}_\text{percentile}^{m,n,k}$ in the $(\mathcal{F}_\text{PFC}^{m,n}$,$ \mathcal{F}_\text{Mill}$,$ \mathcal{F}_\text{AND}$,$ \mathcal{F}_\text{B2A})$ hybrid model against semi-honest probabilistic polynomial time (PPT) adversaries with statistical corruption.
\end{theorem}

\begin{proof}
    To prove Theorem~\ref{theorem: kth}, we construct a PPT simulator $\mathcal{S}$ to simulate the views of the corrupted party.

    $\mathsf{Case}$ $1$: $P_0$ is corrupted.
    In the hybrid model, $P_0$ receives messages from the functionalities, $\mathcal{F}_\text{PFC}^{m,n}, \mathcal{F}_\text{Mill}, \mathcal{F}_\text{AND}, \mathcal{F}_\text{B2A}$.
    By the universal composability security of the functionalities, their outputs are uniformly random.
    We construct the simulator $\mathcal{S}$ to samples uniformly random values to simulate the messages that $P_0$ receives.

    $\mathsf{Case}$ $2$: $P_1$ is corrupted.
    The construction of $\mathcal{S}$ and the indistinguishability are the same as $\mathsf{Case}$ $1$.
\end{proof}

\subsection{The protocol of Secure Counting Max} 
\label{appendix: ex_max}
\myherebox{Protocol $\Pi_\text{CMax}$}{white!20}{white!10}{
   \emph{$\mathsf{Input:}$} 
   For $b \in \{0, 1\}$ $P_b$ inputs the private prefix-sum frequency $\langle \vec{z}\rangle_b \in \mathbb{Z}_{2^{\ell_n}}^m$ and the private $\langle k\rangle_b$. 

   \emph{$\mathsf{Output:}$} 
   $P_b$ outputs the max value $\langle \omega\rangle_b$. 

\underline{\textbf{Protocol:}}
\begin{enumerate}[label=\arabic*), leftmargin=*]
    \item $P_b$ computes $1\{\langle \vec{z}[i]\rangle = n\}$ and gets $\langle \vec{u}'[i]\rangle$, $i \in [m]$.
    \item $P_b$ sets $\langle \vec{u}\rangle[0]:=\langle \vec{u}'[0]\rangle, \langle \vec{u}[i]\rangle:=\langle \vec{u}'[i-1]\rangle \oplus \langle \vec{u}'[i]\rangle$.
    \item $P_b$ invokes $\mathcal{F}_\text{B2A}$ with $\langle \vec{u}\rangle_b$ as input and learns $\langle \vec{t}\rangle_b^{2^{l_m}}$.
    \item $P_b$ sets $\langle \omega\rangle_b := \sum_{j=0}^{m} j\langle \vec{t}[j]\rangle_b^{2^{l_m}}$.
\end{enumerate}
}{The Secure Counting Max Protocol}{algo: exmax}

\end{document}